\documentclass{article}

\usepackage{microtype}
\usepackage{graphicx}
\usepackage{subcaption}
\usepackage{booktabs} 
\usepackage{multirow}

\usepackage{hyperref}

\usepackage[preprint]{icml2026}

\usepackage{amsmath}
\usepackage{amssymb}
\usepackage{mathtools}
\usepackage{amsthm}
\usepackage{bbm}

\usepackage[capitalize,noabbrev]{cleveref}

\usepackage{enumitem}
\usepackage{placeins}

\usepackage{makecell}

\theoremstyle{plain}
\newtheorem{theorem}{Theorem}[section]

\newtheorem{lemma}[theorem]{Lemma}

\theoremstyle{definition}
\newtheorem{definition}[theorem]{Definition}
\newtheorem{assumption}[theorem]{Assumption}
\theoremstyle{remark}
\newtheorem{remark}[theorem]{Remark}

\usepackage[disable,textsize=tiny]{todonotes}

\newcommand{\EE}{\mathbb{E}}
\newcommand{\RR}{\mathbb{R}}

\newcommand{\cS}{\mathcal{S}}
\newcommand{\cA}{\mathcal{A}}

\newcommand{\cO}{\mathcal{O}}

\newcommand{\cB}{\mathcal{B}}

\newcommand{\ba}{\boldsymbol{a}}

\newcommand{\Wass}{\mathcal{W}_1}
\DeclareMathOperator*{\argmax}{arg\,max}

\icmltitlerunning{Mean-Field RL without Synchrony}

\begin{document}

\twocolumn[
    \icmltitle{Mean-Field Reinforcement Learning without Synchrony}


    \icmlsetsymbol{equal}{*}

    \begin{icmlauthorlist}
        \icmlauthor{Shan Yang}{isem}
    \end{icmlauthorlist}

    \icmlaffiliation{isem}{Department of Industrial Systems Engineering and Management, National University of Singapore, Singapore}

    \icmlcorrespondingauthor{Shan Yang}{yang\_shan@u.nus.edu}

    \icmlkeywords{Machine Learning, ICML, Mean-Field RL, Asynchronous MARL, Population State Measure}

    \vskip 0.3in
]



\printAffiliationsAndNotice{}  

\begin{abstract}
Mean-field reinforcement learning (MF-RL) scales multi-agent RL to large populations by reducing each agent's dependence on others to a single summary statistic---the \emph{mean action}. However, this reduction requires every agent to act at every time step; when some agents are idle, the mean action is simply undefined. Addressing asynchrony therefore requires a different summary statistic---one that remains defined regardless of which agents act. The \emph{population distribution} $\mu \in \Delta(\cO)$---the fraction of agents at each observation---satisfies this requirement: its dimension is independent of $N$, and under exchangeability it fully determines each agent's reward and transition. Existing MF-RL theory, however, is built on the mean action and does not extend to $\mu$. We therefore construct the \emph{Temporal Mean Field} (TMF) framework around the population distribution $\mu$ from scratch, covering the full spectrum from fully synchronous to purely sequential decision-making within a single theory. We prove existence and uniqueness of TMF equilibria, establish an $O(1/\sqrt{N})$ finite-population approximation bound that holds regardless of how many agents act per step, and prove convergence of a policy gradient algorithm (TMF-PG) to the unique equilibrium. Experiments on a resource selection game and a dynamic queueing game confirm that TMF-PG achieves near-identical performance whether one agent or all $N$ act per step, with approximation error decaying at the predicted $O(1/\sqrt{N})$ rate.
\end{abstract}

\section{Introduction}
\label{sec:intro}

Mean-field reinforcement learning (MF-RL) \citep{yang2018mean} makes multi-agent RL tractable in large populations: each agent replaces the full joint action with a single summary statistic, the \emph{mean action} $\bar{a}$, averaged over neighboring agents. This reduction, however, rests on a structural assumption that is easy to overlook: every agent must choose an action at every time step, so that $\bar{a}$ is always well-defined. In many real-world domains, such as traffic networks, cloud schedulers, and wireless spectrum markets, the assumption rarely holds---agents routinely act at different times, with some committing to actions while others wait, are in transit, or have not yet arrived. For any agent that does not act at step $t$, its action simply does not exist, and the mean action $\bar{a}$ cannot be formed. Several recent methods address asynchronous execution in MARL---through macro-actions \citep{macdec2020}, delay randomization \citep{ace2024}, communication ordering \citep{seqcomm2023}, or Stackelberg decomposition \citep{steer2024}---but none provides a mean-field formulation, and all still build on synchronous Markov games in which every agent must act before the environment steps.

Any replacement for $\bar{a}$ must satisfy three requirements: it must be defined at every step regardless of which agents act (persistence), its dimension must not grow with the number of agents $N$ (scalability), and it must carry enough information to determine each agent's reward and transition (sufficiency). One natural candidate is the global state: it persists across time steps regardless of which agents act. However, it records every agent's observation individually, so its dimension grows with $N$---violating the scalability requirement. What does satisfy all three is the distribution of agents across observations: every agent $i$ occupies a well-defined observation $o_i \in \cO$ at every step whether or not it acts. Denoting this distribution by $\mu \in \Delta(\cO)$---the \emph{population distribution}---its dimension is $|\cO|$, independent of $N$, and under exchangeability, $\mu$ fully determines each agent's reward and transition. This makes $\mu$ the unique mean-field object that remains well-defined, scalable, and sufficient under any timing of agent decisions.

However, adopting $\mu$ does not by itself resolve the problem: the Bellman equations, equilibrium definitions, and convergence results in MF-RL are all built around $\bar{a}$, so a new theory around $\mu$ must be developed from scratch. We therefore construct a new framework---the \emph{Temporal Mean Field} (TMF) framework---around the population distribution $\mu$. The framework is parameterized by the number of agents that act at each step, ranging from all $N$ (fully synchronous) to just one (purely sequential), and covers the entire spectrum within a single theory. Our contributions:
\begin{enumerate}
    \item \textbf{Unified framework (Section~\ref{sec:framework}).} We formalize how $\mu$ evolves when any subset of agents acts, define each agent's value function conditioned on $\mu$, and introduce an equilibrium concept that couples the policy and the population distribution.
    
    \item \textbf{Theoretical guarantees (Section~\ref{sec:theory}).} We prove existence and uniqueness of the equilibrium (Theorem~\ref{thm:existence_uniqueness}) and establish an $O(1/\sqrt{N})$ finite-population approximation bound that holds regardless of how many agents act per step (Theorem~\ref{thm:n_approx}).
    
    \item \textbf{TMF Reinforcement Learning (Section~\ref{sec:algorithm}).} We propose TMF-PG, a policy gradient algorithm, and prove that it converges to the unique equilibrium (Theorem~\ref{thm:convergence}).
    
    \item \textbf{Experimental validation (Section~\ref{sec:experiments}).} On a resource selection game and a dynamic queueing game, TMF-PG achieves near-identical performance whether one agent or all $N$ act per step, with approximation error decaying at the predicted $O(1/\sqrt{N})$ rate.
\end{enumerate}

\section{Preliminaries}
\label{sec:prelim}

\subsection{Markov Games and Mean-Field RL}
\label{subsec:mg}

A Markov game \citep{littman1994markov} for $N$ exchangeable agents is defined by $\langle \cS, \cO, \cA, P, r, \gamma \rangle$, where $\cS$ is the system (global) state space, $\cO$ the individual observation space, $\cA$ the shared action space, $P$ the transition function, $r$ the reward function, and $\gamma \in [0,1)$ the discount factor. The system state $s_t \in \cS$ captures global environment variables; each agent $i$ receives a local observation $o_t^i \in \cO$ (e.g., its position, assigned task, or service status). Agents are homogeneous: they share the same observation and action spaces and the same reward structure. At each time step, agents choose actions $a_t^i \in \cA$ and the system transitions via $s_{t+1} \sim P(\cdot | s_t, \ba_t)$, where $\ba_t = (a_t^1, \ldots, a_t^N)$ is the joint action.

To scale to large $N$, \citet{yang2018mean} proposed the mean-field approximation:
\begin{equation}
    Q^j(s, a^j, \ba^{-j}) \approx Q^j(s, a^j, \bar{a}^j), \quad \bar{a}^j = \frac{1}{|N(j)|}\sum_{k \in N(j)} a^k,
    \label{eq:yang_mf}
\end{equation}
where $N(j)$ denotes agent $j$'s neighbors. This reduces the exponential dependence on joint actions to a single sufficient statistic $\bar{a}^j$, enabling tractable learning in large populations. Note that under exchangeability the system state is simply the joint observation profile, $s_t = (o_t^1, \ldots, o_t^N)$, so the Q-function's dependence on $s$ is equivalently a dependence on all agents' observations.

\subsection{Synchrony as a Structural Requirement}
\label{subsec:synchrony}

The tractability of the mean-field approximation, however, rests on a structural assumption that is seldom made explicit: all agents must act simultaneously. To see why, we first formalize the notion of a decision protocol.

\begin{definition}[Decision protocol]
\label{def:protocol}
A \emph{decision protocol} $\Pi = \{\cB_t\}_{t=0}^{\infty}$ specifies, for each time step $t$, a subset $\cB_t \subseteq [N]$ of agents that act, with batch size $B_t = |\cB_t|$. A protocol is \emph{synchronous} if $B_t = N$ for all $t$, and \emph{asynchronous} otherwise ($1 \leq B_t \leq N$, possibly varying over time).
\end{definition}

Under any asynchronous protocol ($B_t < N$), agents outside $\cB_t$ do not act at time $t$, so their actions are undefined and the mean action $\bar{a}^j$ in Equation~\ref{eq:yang_mf} cannot be computed. A natural workaround is to substitute \emph{stale actions}---the most recent action of each inactive agent---but this is unsatisfactory. Stale actions introduce a bias that does not vanish with $N$. Moreover, mean-field Q-learning is an iterative scheme in which agents best-respond to the current mean action, which itself reflects the current policy; substituting outdated actions breaks this consistency loop and invalidates the convergence guarantees. The difficulty is fundamental: the mean action is not merely hard to estimate under asynchrony, but \emph{undefined} as a quantity whenever $B_t < N$. Addressing this requires replacing the mean field action $\bar{a}^j$ with a different summary statistic.

\subsection{The Population Distribution as Mean Field}
\label{subsec:pop_dist}

Unlike actions, agent observations persist across time steps: every agent occupies an observation $o_t^i \in \cO$ at every time step, regardless of whether it acts. At any time $t$, we can therefore measure the fraction of agents at each observation $o$:
\begin{equation}
    \mu_t(o) = \frac{1}{N}\sum_{i=1}^N \mathbbm{1}[o_t^i = o],
    \label{eq:mu_def}
\end{equation}
where $\mathbbm{1}[\cdot]$ is the indicator function. Collecting these fractions into a vector gives the \emph{population distribution} $\mu_t = (\mu_t(o))_{o \in \cO}$, which lies in the probability simplex $\Delta(\cO)$. Because $\mu_t$ is defined entirely from observations, it is well-defined under any decision protocol.

Can $\mu_t$ replace the mean action $\bar{a}^j$ as the mean-field object? Under exchangeability and a shared policy $\pi$, knowing $\mu_t$ and $\pi$ determines the aggregate action statistics, so $\mu_t$ in principle carries the information that $\bar{a}^j$ provides. However, the existing MF-RL Bellman equation and Q-iteration are expressed in terms of $\bar{a}^j$, and their convergence and approximation guarantees rely on the simultaneous-move structure. Replacing $\bar{a}^j$ with $\mu_t$ therefore requires new dynamics that describe how $\mu_t$ evolves under partial updates, a new Bellman equation conditioned on $\mu_t$, and a new equilibrium concept---which is the subject of the next section.

\section{The Temporal Mean Field Framework}
\label{sec:framework}

We now present the TMF framework, in which each agent's policy is conditioned on the population distribution $\mu_t$ rather than the mean action $\bar{a}_t$. The framework consists of three components: a dynamic that describes how $\mu_t$ evolves when any subset of agents acts (Section~\ref{subsec:dynamics}), a Bellman equation that defines each agent's value conditioned on $\mu_t$ (Section~\ref{subsec:bellman}), and an equilibrium concept that couples the policy and the population trajectory (Section~\ref{subsec:mfe}).

\subsection{TMF Dynamic}
\label{subsec:dynamics}

As established in Section~\ref{subsec:pop_dist}, exchangeability ensures that each agent's reward and transition depend on other agents only through $\mu_t$. The evolution of the population therefore reduces to tracking how $\mu_t$ changes as agents act.

Suppose all agents follow a shared policy $\pi(a \mid o, \mu)$ that maps each agent's observation and the current population state to a distribution over actions. At step $t$, $B_t$ of the $N$ agents each choose an action $a_t^i \sim \pi(\cdot | o_t^i, \mu_t)$ and transition to $o_{t+1}^i \sim P(\cdot | o_t^i, a_t^i, \mu_t)$; we call these agents \emph{active} ($\cB_t$). The remaining $N - B_t$ \emph{passive} agents transition via $o_{t+1}^j \sim P_0(\cdot | o_t^j, \mu_t)$, where $P_0$ captures environment-driven changes that occur without an explicit action.

Given these transitions, the mean-field approximation yields a deterministic recursion for $\mu_t$ by replacing per-agent contributions with their expectations:
\begin{definition}[TMF Dynamic]
\label{def:tmf_dynamic}
\begin{multline}
    \bar{\mu}_{t+1}(o) = \sum_{o'} \mu_t(o') \bigg[\frac{B_t}{N}\sum_a \pi(a|o',\mu_t)\,P(o|o',a,\mu_t) \\
    + \frac{N-B_t}{N}\,P_0(o|o',\mu_t)\bigg].
    \label{eq:mu_expected}
\end{multline}
\end{definition}
For each current observation $o'$, the bracketed term blends two contributions: a fraction $B_t/N$ of agents act under $\pi$ and transition via $P$, while the remaining $(N - B_t)/N$ transition passively via $P_0$. Definition~\ref{def:tmf_dynamic} defines a deterministic map $\mu_t \mapsto \bar{\mu}_{t+1}$ for any fixed $\pi$. Iterating from $\mu_0$ therefore produces a unique sequence $\{\mu_t\}_{t \geq 0}$, which we call the \emph{population trajectory}.

\begin{remark}
    Definition~\ref{def:tmf_dynamic} unifies all decision protocols through the batch size $B_t$. When $B_t = N$ (synchronous), every agent is active, the $P_0$ term vanishes, and the recursion reduces to the Kolmogorov forward equation. When $B_t = 1$ (sequential), only one agent acts while the remaining $N-1$ transition passively via $P_0$. By varying $B_t$ between these extremes, the TMF dynamic captures synchronous, sequential, and all intermediate protocols.
\end{remark}

\subsection{TMF Bellman Equation}
\label{subsec:bellman}

Given the TMF dynamic, we can define each agent's expected return under a shared policy. When active, each agent selects actions via a shared stochastic policy $\pi(\cdot|o, \mu) \in \Delta(\cA)$, conditioned on its own observation $o$ and the current population distribution $\mu$. Given $\pi$, the TMF dynamic determines the entire future trajectory from any $\mu_t$, so the value function depends only on $(o, \mu_t)$:
\begin{multline}
    V^\pi(o, \mu_t) = \sum_a \pi(a|o, \mu_t) \bigl[r(o, a, \mu_t) \\
    + \gamma \sum_{o'} P(o'|o, a, \mu_t)\, V^\pi(o', \mu_{t+1})\bigr],
    \label{eq:v_pi}
\end{multline}
where $\mu_{t+1}$ follows from Equation~\ref{eq:mu_expected}. Because $V^\pi$ depends on $\mu_t$ rather than on individual agent states, and the TMF dynamic provides a deterministic evolution for $\mu_t$ under any batch size $B_t$, this Bellman equation is well-defined for any active agent regardless of the decision protocol.

\subsection{TMF Equilibrium}
\label{subsec:mfe}

We now define the TMF equilibrium given the TMF dynamic and the Bellman equation above. Because $\mu_t$ evolves according to the policy $\pi$ while $\pi$ is optimized against $\mu_t$, the equilibrium requires not only that $\pi$ is optimal given a population trajectory $\{\mu_t\}$, but also that $\{\mu_t\}$ is the trajectory that $\pi$ itself induces.
\begin{definition}[TMF Equilibrium]
\label{def:mfe}
A policy $\pi^*$ is a \emph{TMF equilibrium policy} if, when $\{\mu_t^*\}$ is the population trajectory generated by $\pi^*$ through Equation~\ref{eq:mu_expected},
\[
    V^{\pi^*}(o, \mu_t^*) \geq V^{\pi}(o, \mu_t^*) \quad \forall\, \pi,\, o,\, t.
\]
\end{definition}

Unlike a Nash equilibrium in the Markov game, which requires only that each agent's strategy is a best response, the TMF equilibrium additionally incorporates the feedback loop between the shared policy  and the aggregate population dynamics. The TMF equilibrium formalizes the self-consistency between the policy $\pi^*$ and the population trajectory $\{\mu_t^*\}$ it induces, and is well-defined for any decision protocol.

Together, the TMF dynamic, the TMF Bellman equation, and the TMF equilibrium form a complete asynchronous mean-field framework. Whether this equilibrium exists, is unique, and approximates the finite-agent game remains to be established.

\section{Theoretical Guarantees for TMF}
\label{sec:theory}

The TMF framework defines an equilibrium as a self-consistent policy-trajectory pair $(\pi^*, \{\mu_t^*\})$, but does not guarantee that such a pair exists or that it is unique. We now establish both properties under mild regularity conditions, and show that the equilibrium approximates the finite $N$-agent game at rate $O(1/\sqrt{N})$, all independent of the decision protocol.

\subsection{Assumptions}

The analysis requires three standard assumptions: Lipschitz continuity and ergodicity of the transition functions, and a displacement monotonicity condition that guarantees uniqueness.

To formalize these assumptions, we first define a distance on $\Delta(\cO)$ to quantify how sensitively the model primitives respond to changes in the population distribution. Let $\Wass(\mu, \nu) = \sum_{o \in \cO} |\mu(o) - \nu(o)|$ denote the $\ell_1$ distance on $\Delta(\cO)$.
\begin{assumption}[Lipschitz continuity]
\label{asmp:lipschitz}
The reward $r$, active transition function $P$, passive transition function $P_0$, and policy $\pi$ are Lipschitz continuous in $\mu$ with respect to $\Wass$:
\begin{align}
    |r(o, a, \mu) - r(o, a, \mu')| &\leq L_r\, \Wass(\mu, \mu'), \\
    \Wass(P(\cdot|o,a,\mu), P(\cdot|o,a,\mu')) &\leq L_P\, \Wass(\mu, \mu'), \\
    \Wass(P_0(\cdot|o,\mu), P_0(\cdot|o,\mu')) &\leq L_{P_0}\, \Wass(\mu, \mu'), \\
    \|\pi(\cdot|\cdot, \mu) - \pi(\cdot|\cdot, \mu')\|_\infty &\leq L_\pi\, \Wass(\mu, \mu'),
\end{align}
where $L_r, L_P, L_{P_0}, L_\pi > 0$ are the respective Lipschitz constants.
\end{assumption}

The $L_\pi$ condition is satisfied by softmax and entropy-regularized policies. In the proofs, $L_\pi$ bounds how much the optimal policy changes when $\mu$ shifts.

Lipschitz continuity bounds how the model primitives respond to changes in $\mu$, but does not ensure that a single transition step contracts the distance between agents at different observations. To quantify this contraction, define the \emph{active transition matrix} $K_{\pi,\mu}(o'|o) = \sum_a \pi(a|o)\,P(o'|o,a,\mu)$ and its \emph{Dobrushin contraction coefficient}
\[
    \rho(K_{\pi,\mu}) \;=\; \max_{o_1 \neq o_2} \frac{1}{2}\sum_{o'} \bigl|K_{\pi,\mu}(o'|o_1) - K_{\pi,\mu}(o'|o_2)\bigr|.
\]
Denote by $\rho_P = \sup_{\mu,\, \pi:\, \pi(a|o)>0\;\forall a,o} \rho(K_{\pi,\mu})$ the worst-case contraction coefficient over all policies and population distributions, and by $\rho_{P_0} = \sup_\mu \rho(P_0(\cdot|\cdot,\mu))$ the analogous coefficient for passive transitions.

\begin{assumption}[Ergodicity of active transitions]
\label{asmp:ergodicity}
The active transition matrix satisfies $\rho_P < 1$.
\end{assumption}

Assumption~\ref{asmp:ergodicity} requires that agents at different observations become strictly closer after one active transition step, regardless of the policy and population distribution. This holds whenever there exists an observation reachable from every other observation under some action---the case in resource selection and scheduling, where every agent can choose any resource. Since $P_0$ is a stochastic matrix, $\rho_{P_0} \leq 1$ holds automatically and requires no additional assumption. Together with Assumption~\ref{asmp:lipschitz}, ergodicity supports the existence and approximation results in Theorems~\ref{thm:existence_uniqueness} and~\ref{thm:n_approx}; uniqueness additionally requires a monotonicity condition.

\begin{assumption}[Discrete monotonicity]
\label{asmp:monotonicity}
For any two population distributions $\mu, \mu' \in \Delta(\cO)$ and their corresponding optimal policies $\pi_\mu$, $\pi_{\mu'}$:
\begin{equation}
    \sum_o (\mu(o) - \mu'(o))\bigl[V_{\pi_\mu}(o, \mu) - V_{\pi_{\mu'}}(o, \mu')\bigr] \geq \eta\, \Wass(\mu, \mu')^2,
\end{equation}
where $\eta > 0$ is the monotonicity constant.
\end{assumption}

Assumption~\ref{asmp:monotonicity} is a discrete monotonicity condition: when two population distributions differ, the corresponding value functions ``push back'' against the difference, preventing multiple equilibria. This is the finite-state counterpart of the Lasry--Lions monotonicity condition \citep{lasry2007mean}, standard for uniqueness in continuous-state mean-field games. It holds naturally in congestion games, resource allocation, and epidemic models, where increased crowding in a state reduces its value.

\subsection{Existence and Uniqueness of the TMF Equilibrium}

A key question is whether the TMF equilibrium is robust to asynchrony: does the equilibrium exist and remain unique regardless of how many agents act per step?

\begin{theorem}[Existence and uniqueness]
\label{thm:existence_uniqueness}
Under Assumptions~\ref{asmp:lipschitz}--\ref{asmp:monotonicity}, if the monotonicity constant satisfies
\begin{equation}
    \eta > L_V + \frac{R_{\max}\, L_\pi}{(1-\gamma)^2},
    \label{eq:contraction_condition_main}
\end{equation}
where $L_V = (L_r + \gamma L_P R_{\max}/(1-\gamma))/(1-\gamma)$ is the value sensitivity constant (Lemma~\ref{lem:value_sensitivity}), then for any batch size $B \in \{1, \ldots, N\}$, a TMF equilibrium $(\pi^*, \{\mu_t^*\})$ exists and is unique.
\end{theorem}
The right-hand side comprises two sources of sensitivity to $\mu$: $L_V$ measures how the value function changes when the population trajectory is perturbed (Lemma~\ref{lem:value_sensitivity}), while $R_{\max}\, L_\pi/(1-\gamma)^2$ captures how the optimal policy shifts in response to the same perturbation. The condition is easier to satisfy when the discount factor $\gamma$ is small or when the transition function depends weakly on $\mu$ (small $L_P$); the precise bound arises from the uniqueness proof (Appendix~\ref{app:proof_existence}).

Uniqueness holds for every batch size $B \in \{1, \ldots, N\}$: regardless of how agents are scheduled, there is exactly one self-consistent policy--trajectory pair, giving agents a well-defined learning target without coordination.

\subsection{Finite-Population Approximation}

The TMF equilibrium is defined in the mean-field limit. To justify its use in finite systems, we bound the \emph{exploitability} $\epsilon_N$: the maximum gain any single agent can achieve by deviating from $\pi^*$ in the $N$-agent game. Writing $\hat{\mu}^N_{\pi}$ for the empirical distribution trajectory when all agents follow $\pi$, the exploitability is $\epsilon_N = \max_i \sup_{\pi^i} [V_{\pi^i}(o_0^i, \hat{\mu}^N_{\pi^i, \pi^{*,-i}}) - V_{\pi^*}(o_0^i, \hat{\mu}^N_{\pi^*})]$.

\begin{theorem}[$N$-agent approximation]
\label{thm:n_approx}
Under Assumptions~\ref{asmp:lipschitz}--\ref{asmp:monotonicity}, let $(\pi^*, \{\mu_t^*\})$ be the TMF equilibrium. For any batch size $B$ such that $\alpha = \tfrac{B}{N}(\rho_P + L_P) + (1 - \tfrac{B}{N})(\rho_{P_0} + L_{P_0}) < 1$:
\begin{equation}
    \epsilon_N \leq \frac{C}{\sqrt{N}},
    \label{eq:n_approx}
\end{equation}
where $C$ depends on $L_r$, $L_P$, $L_{P_0}$, $\rho_P$, $\rho_{P_0}$, $|\cO|$, and $T$, but not on $B$.
\end{theorem}

Theorem~\ref{thm:n_approx} guarantees that the TMF equilibrium is an $O(1/\sqrt{N})$-approximate Nash equilibrium of the finite-agent game for any batch size satisfying $\alpha < 1$. Together with Theorem~\ref{thm:existence_uniqueness}, both the equilibrium and its finite-population guarantee are invariant to the batching protocol, so the batch size can be chosen based on system constraints, provided it satisfies the condition on $\alpha$. The proof is given in Appendix~\ref{app:proof_napprox}.

\begin{remark}[Batch size and passive coupling]
\label{rem:tradeoff}
The condition $\alpha < 1$ links the minimum number of active agents to the strength of passive coupling: weaker coupling allows smaller batches. Since $\alpha$ interpolates linearly between $\rho_P + L_P$ (at $B = N$) and $\rho_{P_0} + L_{P_0}$ (at $B = 0$), increasing $B$ reduces $\alpha$ whenever $\rho_P + L_P < \rho_{P_0} + L_{P_0}$. The condition requires that agents' transitions are not dominated by population-level coupling, which is the regime where mean-field modeling is appropriate. When $\rho_P + L_P < \rho_{P_0} + L_{P_0}$, $\alpha < 1$ is equivalent to
\[
    \frac{B}{N} > \frac{\rho_{P_0} + L_{P_0} - 1}{\rho_{P_0} + L_{P_0} - \rho_P - L_P}.
\]
When idle agents transition independently of $\mu$ ($L_{P_0} \approx 0$, e.g., remaining at their current observation), the lower bound is non-positive and any $B \geq 1$ suffices. Conversely, if the active transition's dependence on $\mu$ is so strong that $\rho_P + L_P \geq 1$, then $\alpha \geq 1$ for all $B$ and the bound does not apply. In most applications passive coupling is moderate, and the condition permits small batches.
\end{remark}

\section{TMF Reinforcement Learning}
\label{sec:algorithm}

Having established the TMF equilibrium and its theoretical properties, we now present a policy gradient algorithm for learning it and prove its convergence.

The TMF equilibrium requires a self-consistent policy-trajectory pair, but improving $\pi$ changes $\{\mu_t\}$, which in turn changes what $\pi$ should be. To resolve this coupling, we propose the \emph{TMF Policy Gradient} (TMF-PG) algorithm, which alternates between data collection and policy improvement:
\begin{itemize}
    \item \textbf{Rollout.} Execute $\pi_\theta$ in the $N$-agent system for $T$ stages, producing the empirical trajectory $\{\hat{\mu}_t\}$ and per-agent rewards $\{r_t^j\}$ directly from interaction (model-free). This is standard RL data collection.
    \item \textbf{Policy update.} Treating $\{\hat{\mu}_t\}$ as fixed, compute advantage estimates $\hat{A}_t^j = \hat{Q}_t^j - \hat{V}_t^j$ and update the policy by gradient ascent:
    \begin{equation}
        \theta \gets \theta + \alpha \sum_{t,j} \nabla_\theta \log \pi_\theta(a_t^j | o_t^j, \hat{\mu}_t)\, \hat{A}_t^j.
        \label{eq:pg_update}
    \end{equation}
    This reduces to a standard single-agent policy gradient given the fixed trajectory $\{\hat{\mu}_t\}$.
\end{itemize}

Each agent's update requires only its own trajectory and $\hat{\mu}_t$, so the per-agent computational cost is independent of $N$. The full pseudocode is given in Algorithm~\ref{alg:fb} (Appendix~\ref{app:algorithm}).

\paragraph{Convergence.} We show that if the policy update at each iteration produces an increasingly accurate best response to the current trajectory---as guaranteed by standard policy gradient convergence---then the alternating iteration converges to the unique TMF equilibrium.

\begin{theorem}[Convergence of TMF-PG]
\label{thm:convergence}
Under the conditions of Theorem~\ref{thm:existence_uniqueness}, suppose the backward pass at each iteration $k$ produces a policy $\pi_k$ satisfying
\[
    \|\pi_k - \pi_{\boldsymbol{\mu}^{(k)}}\|_\infty \leq \epsilon_k, \quad \epsilon_k \to 0,
\]
where $\pi_{\boldsymbol{\mu}^{(k)}}$ denotes the exact best response to the current trajectory $\boldsymbol{\mu}^{(k)}$. Then the TMF-PG iterates converge to the unique TMF equilibrium:
\[
    \max_t \Wass(\mu_t^{(k)}, \mu_t^*) \to 0 \quad \text{as } k \to \infty,
\]
and correspondingly $\|\pi_k - \pi^*\|_\infty \to 0$. The convergence holds for any batch size $B$ satisfying the condition of Theorem~\ref{thm:n_approx}.
\end{theorem}

The condition $\epsilon_k \to 0$ is achievable in practice: for a fixed trajectory $\{\hat{\mu}_t\}$, each agent faces a finite-horizon MDP, so running sufficiently many policy gradient steps per iteration drives the policy toward the best response. The outer loop converges because the rollout--best-response alternation is contracting under the monotonicity conditions of Theorem~\ref{thm:existence_uniqueness}, and the approximation errors $\epsilon_k$ are absorbed as they vanish. The proof is given in Appendix~\ref{app:proof_convergence}.

\section{Related Work}
\label{sec:related}

\paragraph{Mean-field MARL.} The mean-field approach to MARL was introduced by \citet{yang2018mean}, who used the mean action $\bar{a}$ as the mean field. \citet{subramanian2019reinforcement} subsequently used the empirical state-action distribution as the mean-field object, moving beyond $\bar{a}$, but still under synchronous dynamics. Extensions along the spatial dimension include multi-type MF-MARL \citep{subramanian2020multi}, which partitions agents by type; major-minor MF-MARL \citep{cui2024m3fc}, which distinguishes influential ``major'' agents from a homogeneous ``minor'' population; and graph-structured MFG approaches using graphon \citep{graphon2024} and graphex \citep{graphex2024} limits. All these works assume synchronous decision-making. Our framework opens the orthogonal temporal dimension by removing this assumption entirely.

\paragraph{Asynchronous MARL.} Several works address asynchronous execution, but all remain on top of synchronous game models. MacDec-POMDP \citep{macdec2020} adds temporally extended macro-actions to Dec-POMDP; ACE \citep{ace2024} applies delay randomization to MAPPO; SeqComm \citep{seqcomm2023} introduces priority-based communication ordering; STEER \citep{steer2024} decomposes within-step action selection into a Stackelberg hierarchy. These methods engineer around the consequences of asynchrony---variable durations, delays, sequential ordering---but the underlying state transition remains synchronous: all agents must act before the environment advances. Our framework removes this requirement at the model level, providing a mean-field formulation---dynamics, Bellman equation, and equilibrium---that is natively defined for any number of active agents per step, without an underlying synchronous substrate.

\paragraph{Mean-field game theory.} The population distribution $\mu_t$ is a central object in continuous-time MFG theory \citep{lasry2007mean, huang2006large, carmona2018probabilistic}, but in a fundamentally different setting: continuous time, infinite populations, and PDE-based analysis. Discrete-time MFG theory \citep{gomes2010discrete, saldi2018markov} and its learning algorithms---fictitious play \citep{perrin2020fictitious}, Q-iteration \citep{anahtarci2023learning}, policy gradient \citep{carmona2019linear}---all assume synchronous dynamics. Our work addresses a different regime---discrete time, finite populations, asynchronous decision-making---using model-free RL.

\paragraph{Approximation theory for MARL.} The approximation of finite-$N$ games by mean-field limits has been studied through propagation of chaos results \citep{carmona2018probabilistic, lacker2020convergence}. \citet{tractable2024} show that mean-field RL can be PPAD-hard without structural assumptions. Both lines of work assume synchronous dynamics. Our $O(1/\sqrt{N})$ bound (Theorem~\ref{thm:n_approx}) establishes the first finite-population approximation guarantee under asynchronous dynamics.

\section{Experiments}
\label{sec:experiments}

We validate the theoretical findings on two environments: a \emph{Sequential Resource Selection Game} (SRSG) for clean theoretical verification, and a \emph{Dynamic Queueing Game} (DQG) that demonstrates practical value in a multi-step asynchronous setting. All experiments use 40 independent seeds; error bars and $\pm$ denote one standard deviation.

\subsection{Sequential Resource Selection Game (SRSG)}
\label{subsec:srsg}

To validate robustness to batch size and $O(1/\sqrt{N})$ concentration, we design a congestion game in which agents arrive in batches, each choosing one of several resources whose reward decreases with congestion. The batch size $B$ can be varied from $1$ (fully sequential) to $N$ (fully synchronous).

We consider a setting where $N$ agents arrive in batches of size $B$ and irrevocably choose among $M = 5$ resources over $T = \lceil N/B \rceil$ steps. Agent~$i$'s reward upon choosing resource~$m$ is $r_m(\mu) = v_m - \alpha\, \mu[m]^2$, where $v = (0.5, 0.75, 1.0, 1.25, 1.5)$, $\alpha = 1.0$, and $\mu[m]$ is the fraction of agents already allocated to resource~$m$. The quadratic congestion penalizes overloaded resources, creating a tension between attractive high-value resources and the cost of congestion.

We evaluate all policies by \emph{per-agent welfare}---the average reward across agents---which reflects how efficiently the population spreads across resources: higher welfare means less congestion waste. For reference, welfare equals $0.5$ when all agents crowd onto the best resource ($v_5 - \alpha = 0.5$) and approaches $1.5$ in the uncongested limit. We train the policy using TMF-PG (Algorithm~\ref{alg:fb}) and compare against a \textbf{Myopic} policy that maximizes the instantaneous reward $r_m(\mu)$ without anticipating future arrivals. Existing mean-field RL methods require synchronous updates and do not apply here; Myopic ablates the forward planning that distinguishes TMF-PG, so the performance gap isolates the benefit of anticipating the population trajectory.

\paragraph{Experiment 1: Robustness across batch sizes (Figure~\ref{fig:b_sweep}).}
We fix $N = 100$ and vary the batch size $B \in \{1, 2, 5, 10, 25, 50, 100\}$, spanning from fully sequential ($B = 1$) to fully synchronous ($B = N$). TMF-PG achieves welfare between $1.109$ and $1.158$ at every $B$---consistently above $70\%$ of the congestion-free maximum---with less than $5\%$ variation across the entire range. The policy performs equally well whether agents decide one at a time or all at once. By contrast, Myopic welfare depends heavily on $B$: at $B = 1$, agents arrive one at a time and naturally spread across resources, yielding welfare $1.107$ comparable to TMF-PG; at $B = 100$, all agents choose simultaneously, each selects the highest-value resource, and welfare collapses to $0.500$.

\begin{figure}[t]
    \centering
    \includegraphics[width=0.80\columnwidth]{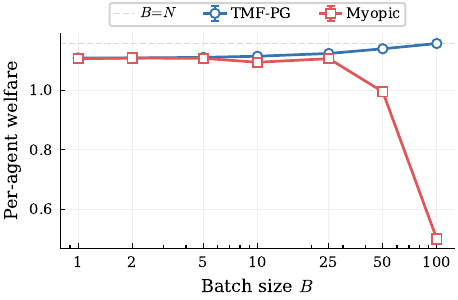}
    \caption{SRSG ($N = 100$): per-agent welfare vs.\ batch size $B$. TMF-PG achieves above $70\%$ of the congestion-free maximum at every $B$; Myopic matches at $B = 1$ but collapses to $0.5$ when all agents choose simultaneously.}
    \label{fig:b_sweep}
\end{figure}

\paragraph{Experiment 2: Finite-population accuracy (Figure~\ref{fig:concentration}).}
TMF-PG is trained on an infinite-population model; how accurate is this approximation for finite $N$? We fix $B = 1$ (fully sequential) and sweep $N \in \{10, 20, 50, 100, 200, 500, 1000, 2000\}$, measuring two indicators. The first is the \emph{cross-run welfare standard deviation}: we run $100$ simulations with the same policy and record how much per-agent welfare fluctuates---smaller fluctuation means the infinite-population prediction is more reliable. The second is the \emph{trajectory prediction error}: the $L_1$ distance between the deterministic population trajectory predicted by Equation~\ref{eq:mu_expected} and the realized $\mu_t$ from the $N$-agent simulation. Both shrink rapidly with $N$: welfare standard deviation drops from $0.016$ ($N = 10$) to $0.001$ ($N = 2000$), and prediction error drops from $0.019$ ($N = 50$) to $0.003$ ($N = 2000$), meaning the model tracks the true dynamics to within $0.3\%$. On log-log axes, both decay at least as fast as $O(1/\sqrt{N})$, consistent with the upper bound of Theorem~\ref{thm:n_approx}. Additional results in Appendix~\ref{app:experiments_detail} (Figure~\ref{fig:prediction_heatmap}) confirm that the prediction error remains stable across batch sizes.

\begin{figure}[t]
    \centering
    \begin{subfigure}[b]{0.48\columnwidth}
        \centering
        \includegraphics[width=\textwidth]{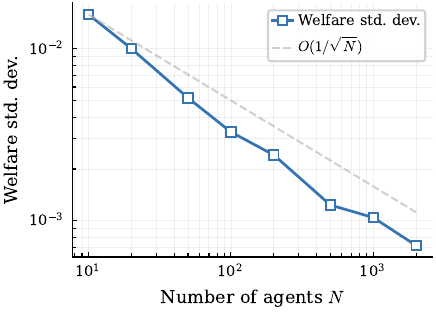}
        \caption{Welfare standard deviation}
        \label{fig:concentration_a}
    \end{subfigure}
    \hfill
    \begin{subfigure}[b]{0.48\columnwidth}
        \centering
        \includegraphics[width=\textwidth]{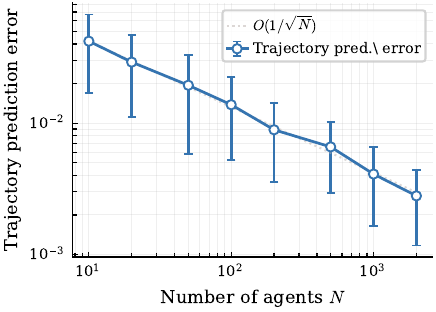}
        \caption{Trajectory prediction error}
        \label{fig:concentration_b}
    \end{subfigure}
    \caption{SRSG ($B = 1$): both welfare standard deviation \textbf{(a)} and trajectory prediction error \textbf{(b)} decay as $O(1/\sqrt{N})$ (dashed reference line), consistent with Theorem~\ref{thm:n_approx}.}
    \label{fig:concentration}
\end{figure}

Taken together, Experiments 1 and 2 validate the two central predictions of the TMF framework under the contraction condition: performance and approximation quality are both independent of the batch size $B$, and the finite-population error vanishes at rate $O(1/\sqrt{N})$.

\subsection{Dynamic Queueing Game (DQG)}
\label{subsec:dqg}

Many resource allocation problems involve repeated decisions over time: an agent who joins a queue today is unavailable for reassignment until service completes. To test whether TMF-PG can exploit such temporal structure, we design a DQG in which agents repeatedly choose servers, wait in queues, and return for reassignment, so that current allocation decisions shape future queue states.

We consider a DQG with $N$ agents choosing among $M = 5$ heterogeneous servers over a horizon of $H = 80$ steps. Each server has a quality, a service rate, and a capacity threshold; when the number of queued agents exceeds capacity, the service rate drops sharply (multiplicative cliff factor $\kappa = 0.2$). Upon completing service, an agent collects a quality-dependent reward and re-enters the free pool. Because only agents who have finished service choose at each step, the effective batch size $B$ varies with the queue dynamics---a natural instance of variable $B$ that the TMF framework handles without modification. The servers are designed with a ``honey-trap'': the highest-quality server (quality $= 3.0$, rate $= 3.0$) has the smallest capacity ($c = 3$), so agents drawn to its high reward risk triggering the cliff penalty once the queue fills.

We evaluate by \emph{per-agent reward}---total reward summed over all agents and all time steps, divided by $N$---the multi-step analogue of per-agent welfare in the SRSG. We compare TMF-PG against the same Myopic baseline, which greedily selects the server with the highest immediate expected reward given current loads. In this setting, Myopic systematically overloads the honey-trap server; TMF-PG anticipates the collective load trajectory and spreads allocation to avoid congestion collapse.

\paragraph{Experiment 3: Value of forward planning in DQG (Figure~\ref{fig:dqg}).}
We sweep $N \in \{50, 100, 150\}$ with capacities scaled linearly ($c_m \propto N/50$). TMF-PG achieves $10$--$30\%$ higher per-agent reward by distributing load away from the honey-trap server before its queue overflows: at $N = 50$, per-agent reward is $8.29$ (TMF-PG) vs.\ $6.38$ (Myopic); at $N = 100$, $4.59$ vs.\ $3.80$. The gap narrows to $+10\%$ at $N = 150$. With linearly scaled capacities, the cliff threshold grows in proportion to $N$ (e.g., the honey-trap server holds $c = 9$ agents instead of $3$), so each agent's marginal contribution to congestion shrinks as $O(1/N)$ and even Myopic rarely triggers the cliff. Sensitivity analysis (Appendix~\ref{app:experiments_detail}, Figure~\ref{fig:dqg_sensitivity}) confirms that the advantage is robust across cliff intensities ($\kappa \in [0.05, 0.5]$) and horizons ($H \in [30, 120]$).

\begin{figure}[t]
    \centering
    \includegraphics[width=0.80\columnwidth]{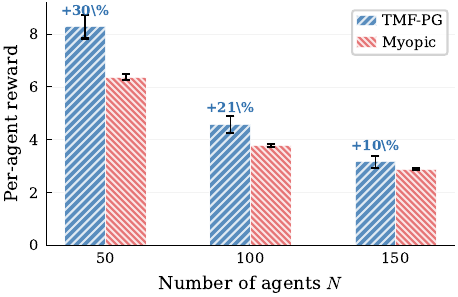}
    \caption{DQG: per-agent reward vs.\ $N$. TMF-PG avoids the cliff penalty on the honey-trap server by anticipating collective load, yielding $10$--$30\%$ higher reward than Myopic.}
    \label{fig:dqg}
\end{figure}

\section{Conclusion}
\label{sec:conclusion}

In this work, we identified a structural limitation of mean-field RL: the mean action, the standard summary statistic, is undefined whenever some agents do not act, making existing MF-RL methods inapplicable to asynchronous or sequential decision-making. To address this, we constructed the TMF framework from scratch around the population distribution $\mu$---the only mean-field object that remains well-defined, scalable, and sufficient regardless of which agents act. We formalized the TMF dynamic, Bellman equation, and equilibrium concept; proved existence and uniqueness of the equilibrium (Theorem~\ref{thm:existence_uniqueness}); established an $O(1/\sqrt{N})$ finite-population approximation bound that holds regardless of how many agents act per step (Theorem~\ref{thm:n_approx}); and proposed TMF-PG, a policy gradient algorithm that provably converges to the unique equilibrium (Theorem~\ref{thm:convergence}). Experiments on a resource selection game and a dynamic queueing game confirm that TMF-PG achieves near-identical performance whether one agent or all $N$ act per step, with approximation error decaying at the predicted $O(1/\sqrt{N})$ rate.

\paragraph{Relation to existing practice.} It is worth noting that several recent asynchronous MARL methods already rely, implicitly, on population-state information: ACE~\citep{ace2024} conditions on delayed joint observations, MacDec-POMDP~\citep{macdec2020} aggregates agent histories through centralized critics, and sequential decision frameworks such as STEER~\citep{steer2024} and SeqComm~\citep{seqcomm2023} condition on predecessors' actions or communication messages that implicitly encode distributional information. These mechanisms can be viewed as heuristic approximations to the population distribution. Our framework provides a principled explanation for \emph{why} such approaches work: the population distribution is a sufficient statistic for mean-field interaction under exchangeability, and the equilibrium it induces is robust to the batching protocol. In this sense, TMF offers theoretical grounding for a design pattern that practitioners have already adopted, while also identifying the precise conditions under which the approach is guaranteed to succeed.

\paragraph{Future directions.} Several extensions are natural. First, combining the population-measure approach (temporal dimension) with structured interaction models such as graphon or graphex limits (spatial dimension) would yield a fully general mean-field framework for networked systems. Second, extending the theory to partial observability---where agents estimate $\mu_t$ from local information rather than observing it directly---would broaden applicability. Third, relaxing exchangeability to accommodate heterogeneous agent types within the TMF dynamic would extend the framework to settings where agents differ in capabilities or objectives.

\bibliographystyle{icml2026}
\bibliography{main}

\newpage
\appendix
\onecolumn
\allowdisplaybreaks

\section{Algorithm}
\label{app:algorithm}

Algorithm~\ref{alg:fb} presents the full pseudocode of TMF-PG. Each iteration consists of a rollout (executing the current policy in the $N$-agent system to collect trajectories and observe $\{\hat{\mu}_t\}$) followed by a policy update (computing returns and updating the policy by policy gradient). In the convergence analysis, these two phases correspond to the forward map $\mathrm{Fwd}$ and the backward map $\mathrm{BR}$ of the operator $\Gamma$.

\begin{algorithm}[H]
\caption{TMF Policy Gradient (TMF-PG)}
\label{alg:fb}
\begin{algorithmic}[1]
\REQUIRE Initial policy parameters $\theta_0$, initial distribution $\mu_0$, batch size $B$, horizon $T$, number of iterations $K$, learning rate $\alpha$
\FOR{$k = 0, 1, \ldots, K-1$}
    \STATE \textbf{// Rollout: execute policy and collect trajectories}
    \STATE Initialize $o_0^i \sim \mu_0$ for all agents $i = 1, \ldots, N$
    \STATE $\hat{\mu}_0(o) \gets \frac{1}{N}\sum_{i=1}^N \mathbbm{1}[o_0^i = o]$
    \FOR{$t = 0, 1, \ldots, T-1$}
        \STATE Sample batch $\cB_t \subseteq \{1, \ldots, N\}$ with $|\cB_t| = B$
        \FOR{$i \in \cB_t$}
            \STATE Sample $a_t^i \sim \pi_{\theta_k}(\cdot | o_t^i, \hat{\mu}_t)$
            \STATE Observe reward $r_t^i$ and next observation $o_{t+1}^i \sim P(\cdot | o_t^i, a_t^i, \hat{\mu}_t)$
        \ENDFOR
        \FOR{$j \notin \cB_t$}
            \STATE Transition $o_{t+1}^j \sim P_0(\cdot | o_t^j, \hat{\mu}_t)$; set $r_t^j \gets 0$
        \ENDFOR
        \STATE $\hat{\mu}_{t+1}(o) \gets \frac{1}{N}\sum_{i=1}^N \mathbbm{1}[o_{t+1}^i = o]$
    \ENDFOR
    \STATE \textbf{// Policy update: compute returns and update policy}
    \STATE Compute returns $G_t^i = \sum_{s \geq t} \gamma^{s-t} r_s^i$ for all agents $i$ and stages $t$
    \STATE Compute advantage estimates $\hat{A}_t^i$ (e.g., via baseline $V_\psi(o, \hat{\mu}_t)$)
    \STATE Update: $\theta_{k+1} \gets \theta_k + \alpha \sum_{t,i} \nabla_{\theta} \log \pi_{\theta_k}(a_t^i | o_t^i, \hat{\mu}_t)\, \hat{A}_t^i$
\ENDFOR
\STATE \textbf{return} $\pi_{\theta_K}$
\end{algorithmic}
\end{algorithm}

\paragraph{Computational complexity.} The rollout simulates $N$ agents for $T$ stages, requiring $O(NT)$ transitions per iteration ($BT$ active transitions via $P$ and $(N-B)T$ passive transitions via $P_0$). The policy update computes returns and policy gradients in $O(NT)$ time. The per-agent complexity is independent of $N$ because the mean-field approximation decouples the agents: each agent's update depends on its own trajectory and the observed population distribution $\hat{\mu}_t$, not on the states of all other agents individually.

\section{Value Function Sensitivity}
\label{app:value_sensitivity}

The following Lipschitz bound on the value function with respect to the population trajectory is used in the proofs of both Theorem~\ref{thm:existence_uniqueness} and Theorem~\ref{thm:n_approx}.

\begin{lemma}[Value sensitivity]
\label{lem:value_sensitivity}
Under Assumption~\ref{asmp:lipschitz}, for any fixed policy $\pi$ and two measure trajectories $\{\mu_t\}$, $\{\mu_t'\}$ with $\max_t \Wass(\mu_t, \mu_t') = d$:
\begin{equation}
    \|V_{\pi}(\cdot, \mu) - V_{\pi}(\cdot, \mu')\|_\infty \leq L_V\, d, \quad \text{where } L_V = \frac{L_r + \gamma L_P R_{\max}/(1-\gamma)}{1-\gamma}.
    \label{eq:value_sensitivity}
\end{equation}
\end{lemma}

\begin{proof}
For a fixed policy $\pi$, the Bellman equation gives:
\[
    V_\pi(o, \mu_t) = \sum_a \pi(a|o)\Bigl[r(o, a, \mu_t) + \gamma \sum_{o'} P(o'|o,a,\mu_t)\, V_\pi(o', \mu_{t+1})\Bigr].
\]
Taking the difference $V_\pi(o, \mu_t) - V_\pi(o, \mu_t')$ and adding and subtracting $P(o'|o,a,\mu_t)\, V_\pi(o', \mu_{t+1}')$ inside the sum:
\begin{align}
    V_\pi(o, \mu_t) - V_\pi(o, \mu_t') &= \sum_a \pi(a|o)\Bigl[\underbrace{r(o,a,\mu_t) - r(o,a,\mu_t')}_{\text{(i) reward difference}} \notag\\
    &\quad + \gamma \sum_{o'} \underbrace{\bigl[P(o'|o,a,\mu_t) - P(o'|o,a,\mu_t')\bigr]\, V_\pi(o', \mu_{t+1}')}_{\text{(ii) transition difference}} \notag\\
    &\quad + \gamma \sum_{o'} P(o'|o,a,\mu_t)\, \underbrace{\bigl[V_\pi(o', \mu_{t+1}) - V_\pi(o', \mu_{t+1}')\bigr]}_{\text{(iii) recursive term}} \Bigr]. \notag
\end{align}
Taking absolute values and bounding each term:

\emph{Term~(i): reward difference.} By Assumption~\ref{asmp:lipschitz}:
\[
    |r(o,a,\mu_t) - r(o,a,\mu_t')| \leq L_r\, \Wass(\mu_t, \mu_t') \leq L_r\, d.
\]

\emph{Term~(ii): transition difference.} We first apply H\"older's inequality ($\sum_o |a_o\, b_o| \leq \|a\|_1\, \|b\|_\infty$) to factor out $V$:
\[
    \Bigl|\sum_{o'} \bigl[P(o'|o,a,\mu_t) - P(o'|o,a,\mu_t')\bigr] V_\pi(o', \mu_{t+1}')\Bigr|
    \leq \|V_\pi\|_\infty \sum_{o'} \bigl|P(o'|o,a,\mu_t) - P(o'|o,a,\mu_t')\bigr|.
\]
On a discrete state space, $\sum_{o'} |P(o') - P'(o')|$ is the $L_1$ distance between two distributions, which equals $\Wass\bigl(P(\cdot|o,a,\mu_t),\, P(\cdot|o,a,\mu_t')\bigr)$. By Assumption~\ref{asmp:lipschitz} ($L_P$-Lipschitz):
\[
    \Wass\bigl(P(\cdot|o,a,\mu_t),\, P(\cdot|o,a,\mu_t')\bigr) \leq L_P\, \Wass(\mu_t, \mu_t') \leq L_P\, d.
\]
So term~(ii) is bounded by $\gamma\, \|V_\pi\|_\infty\, L_P\, d$.

\emph{Term~(iii): recursive term.} Denoting $\delta_t = \|V_\pi(\cdot, \mu_t) - V_\pi(\cdot, \mu_t')\|_\infty$, this contributes $\gamma\, \delta_{t+1}$.

\emph{Combining.} Summing the three contributions:
\[
    \delta_t \leq L_r\, d + \gamma\, \|V_\pi\|_\infty\, L_P\, d + \gamma\, \delta_{t+1}.
\]
Since $\|V_\pi\|_\infty \leq R_{\max}/(1-\gamma)$ by the geometric series of discounted rewards (where $R_{\max} = \max_{o,a,\mu} |r(o,a,\mu)|$), substituting and unrolling the recursion with $\gamma < 1$:
\[
    \delta_t \leq \bigl(L_r + \gamma L_P R_{\max}/(1-\gamma)\bigr)\, d + \gamma\, \delta_{t+1}
    \leq \frac{L_r + \gamma L_P R_{\max}/(1-\gamma)}{1-\gamma}\, d = L_V\, d. \qedhere
\]
\end{proof}

\section{Proof of Theorem~\ref{thm:existence_uniqueness}}
\label{app:proof_existence}

\noindent\textbf{Theorem~\ref{thm:existence_uniqueness}} (Existence and uniqueness).\quad
\emph{Under Assumptions~\ref{asmp:lipschitz}--\ref{asmp:monotonicity} and the contraction condition (Equation~\ref{eq:contraction_condition_main}), for any batch size $B \in \{1, \ldots, N\}$, a TMF equilibrium $(\pi^*, \{\mu_t^*\})$ exists and is unique.}

\noindent\textit{Proof.}\quad We first define the forward-backward operator underlying the proof, then establish existence (via Brouwer) and uniqueness (via contraction).

\subsection{The Forward-Backward Operator} Define $\Gamma: \Delta(\cO)^T \to \Delta(\cO)^T$ as the operator that maps a distribution trajectory $\{\mu_t\}_{t=0}^T$ to a new trajectory $\Gamma(\{\mu_t\}) = \{\mu_t'\}$ via:
\begin{align}
    \pi_\mu(\cdot|o) &= \argmax_\pi \sum_a \pi(a|o)\bigl[r(o,a,\mu) + \gamma \sum_{o'} P(o'|o,a,\mu)\, V_\mu(o')\bigr], \label{eq:phi_backward}\\
    \mu_{t+1}'(o') &= \sum_o \mu_t'(o)\Bigl[\tfrac{B}{N}\sum_a \pi_\mu(a|o)\, P(o'|o,a,\mu_t) + \bigl(1-\tfrac{B}{N}\bigr) P_0(o'|o,\mu_t)\Bigr], \quad \mu_0' = \mu_0, \label{eq:phi_forward}
\end{align}
where $V_\mu$ is the value function obtained by solving the Bellman equation (Equation~\ref{eq:v_pi}) given $\{\mu_t\}$, and $\pi_\mu$ is the corresponding optimal policy. By Assumption~\ref{asmp:lipschitz}, $\pi_\mu$ is $L_\pi$-Lipschitz in $\mu$. A fixed point of $\Gamma$ corresponds to a TMF equilibrium $(\pi^*, \{\mu_t^*\})$.

The operator $\Gamma$ decomposes as $\Gamma = \mathrm{Fwd} \circ \mathrm{BR}$, where $\mathrm{BR}\colon \boldsymbol{\mu} \mapsto \pi_{\boldsymbol{\mu}}$ is the best-response map (Equation~\ref{eq:phi_backward}) and $\mathrm{Fwd}\colon \pi \mapsto \boldsymbol{\mu}'$ evolves the population distribution under $\pi$ via the TMF dynamic (Equation~\ref{eq:phi_forward}). We refer to these as the \emph{backward pass} and \emph{forward pass} throughout the proofs.

\subsection{Existence via Brouwer's Fixed-Point Theorem}
The operator $\Gamma$ satisfies the conditions of Brouwer's fixed-point theorem:
\begin{enumerate}
    \item \textbf{Compact convex domain.} $\Delta(\cO)^T$ is a compact convex subset of $\RR^{|\cO| \times T}$ (since $\cO$ is finite).
    \item \textbf{Self-mapping.} The TMF dynamic (Equation~\ref{eq:mu_expected}) preserves the probability simplex at each step, so $\Gamma$ maps $\Delta(\cO)^T$ into itself.
    \item \textbf{Continuity.} $V_\mu$ is continuous in $\mu$ by Assumption~\ref{asmp:lipschitz} (Lipschitz continuity of $r$ and $P$); $\pi_\mu$ is continuous in $\mu$ by Assumption~\ref{asmp:lipschitz} ($L_\pi$-Lipschitz); $\mathrm{Fwd}$ is continuous by the Lipschitz property of $P$ and $P_0$.
\end{enumerate}
Therefore $\Gamma$ has at least one fixed point---a TMF equilibrium exists.

\subsection{Uniqueness via Monotonicity}

Suppose for contradiction that two distinct TMF equilibria exist: $(\pi, \{\mu_t\})$ and $(\pi', \{\mu_t'\})$ with $d = \max_t \Wass(\mu_t, \mu_t') > 0$. Both are fixed points of $\Gamma$ (Equations~\ref{eq:phi_backward}--\ref{eq:phi_forward}), so at each equilibrium the policy is optimal given the trajectory, and the trajectory is self-consistent under the policy.

\emph{Step 1: Monotonicity lower bound.} Let $\bar{t} = \arg\max_t \Wass(\mu_t, \mu_t')$, so that $\Wass(\mu_{\bar{t}}, \mu_{\bar{t}}') = d$. By Assumption~\ref{asmp:monotonicity}:
\begin{equation}
    \eta\, d^2 \leq \sum_o \bigl(\mu_{\bar{t}}(o) - \mu_{\bar{t}}'(o)\bigr)\bigl[V_{\pi}(o, \mu_{\bar{t}}) - V_{\pi'}(o, \mu_{\bar{t}}')\bigr].
    \label{eq:mono_lower}
\end{equation}

\emph{Step 2: H\"older upper bound.} By H\"older's inequality ($\sum_o |a_o\, b_o| \leq \|a\|_1\, \|b\|_\infty$), the right-hand side of Equation~\ref{eq:mono_lower} satisfies:
\[
    \text{RHS of Equation~\ref{eq:mono_lower}} \leq d \cdot \|V_{\pi}(\cdot, \mu_{\bar{t}}) - V_{\pi'}(\cdot, \mu_{\bar{t}}')\|_\infty.
\]
Dividing by $d > 0$:
\begin{equation}
    \eta\, d \leq \|V_{\pi}(\cdot, \mu_{\bar{t}}) - V_{\pi'}(\cdot, \mu_{\bar{t}}')\|_\infty.
    \label{eq:mono_holder}
\end{equation}

\emph{Step 3: Value sensitivity upper bound.} We bound the value difference by decomposing it via triangle inequality. Adding and subtracting $V_{\pi}(\cdot, \mu_{\bar{t}}')$:
\begin{equation}
    \|V_{\pi}(\cdot, \mu_{\bar{t}}) - V_{\pi'}(\cdot, \mu_{\bar{t}}')\|_\infty \leq \underbrace{\|V_{\pi}(\cdot, \mu_{\bar{t}}) - V_{\pi}(\cdot, \mu_{\bar{t}}')\|_\infty}_{\text{(I) trajectory sensitivity}} + \underbrace{\|V_{\pi}(\cdot, \mu_{\bar{t}}') - V_{\pi'}(\cdot, \mu_{\bar{t}}')\|_\infty}_{\text{(II) policy sensitivity}}.
    \label{eq:triangle_decomp}
\end{equation}

\emph{Bounding term~(I).} This is the value sensitivity under the same policy $\pi$ but different trajectories. By Lemma~\ref{lem:value_sensitivity}:
\[
    \|V_{\pi}(\cdot, \mu_{\bar{t}}) - V_{\pi}(\cdot, \mu_{\bar{t}}')\|_\infty \leq L_V\, d.
\]

\emph{Bounding term~(II).} This is the value difference under the same trajectory $\{\mu_t'\}$ but different policies. Since $\pi'$ is optimal for $\{\mu_t'\}$, we have $V_{\pi'}(o, \mu') \geq V_{\pi}(o, \mu')$ for all $o$. We now bound this gap.

For a fixed trajectory $\{\mu_t'\}$, write the Bellman equations for both policies:
\begin{align*}
    V_{\pi'}(o) &= \sum_a \pi'(a|o)\Bigl[r(o,a,\mu') + \gamma \sum_{o'} P(o'|o,a,\mu')\, V_{\pi'}(o')\Bigr], \\
    V_{\pi}(o)  &= \sum_a \pi(a|o)\Bigl[r(o,a,\mu') + \gamma \sum_{o'} P(o'|o,a,\mu')\, V_{\pi}(o')\Bigr].
\end{align*}
Subtracting, and adding and subtracting $\sum_a \pi(a|o)[\cdots V_{\pi'}(o')]$ to separate the policy difference from the value difference:
\begin{align}
    V_{\pi'}(o) - V_{\pi}(o)
    &= \sum_a \bigl(\pi'(a|o) - \pi(a|o)\bigr)\Bigl[r(o,a,\mu') + \gamma \sum_{o'} P(o'|o,a,\mu')\, V_{\pi'}(o')\Bigr] \notag\\
    &\quad + \gamma \sum_a \pi(a|o) \sum_{o'} P(o'|o,a,\mu')\bigl[V_{\pi'}(o') - V_{\pi}(o')\bigr].
    \label{eq:bellman_diff}
\end{align}

\emph{First term (policy difference).} The bracketed expression satisfies:
\[
    \bigl|r(o,a,\mu') + \gamma \sum_{o'} P(o'|o,a,\mu')\, V_{\pi'}(o')\bigr| \leq R_{\max} + \gamma\, \|V_{\pi'}\|_\infty \leq R_{\max} + \frac{\gamma R_{\max}}{1-\gamma} = \frac{R_{\max}}{1-\gamma}.
\]
So the first term is bounded by $\|\pi' - \pi\|_\infty \cdot R_{\max}/(1-\gamma)$.

\emph{Second term (recursive).} Let $\Delta = \|V_{\pi'} - V_{\pi}\|_\infty$. Since $|V_{\pi'}(o') - V_{\pi}(o')| \leq \Delta$ for all $o'$, and $\pi(\cdot|o)$ and $P(\cdot|o,a,\mu')$ are probability distributions:
\[
    \gamma \Bigl|\sum_a \pi(a|o) \sum_{o'} P(o'|o,a,\mu')\bigl[V_{\pi'}(o') - V_{\pi}(o')\bigr]\Bigr|
    \leq \gamma \underbrace{\sum_a \pi(a|o) \sum_{o'} P(o'|o,a,\mu')}_{= 1} \cdot \Delta = \gamma\, \Delta.
\]

\emph{Combining.} Substituting these bounds into Equation~\ref{eq:bellman_diff} and taking $\sup_o$ (the LHS becomes $\sup_o |V_{\pi'}(o) - V_{\pi}(o)| = \Delta$):
\[
    \Delta \leq \|\pi' - \pi\|_\infty \cdot \frac{R_{\max}}{1-\gamma} + \gamma\, \Delta
    \quad\Longrightarrow\quad
    (1-\gamma)\,\Delta \leq \|\pi' - \pi\|_\infty \cdot \frac{R_{\max}}{1-\gamma}
    \quad\Longrightarrow\quad
    \Delta \leq \frac{R_{\max}}{(1-\gamma)^2}\,\|\pi' - \pi\|_\infty.
\]
By Assumption~\ref{asmp:lipschitz}, $\|\pi - \pi'\|_\infty \leq L_\pi\, d$. So term~(II) of Equation~\ref{eq:triangle_decomp} satisfies:
\[
    \text{(II)} = \Delta \leq \frac{R_{\max}}{(1-\gamma)^2}\, L_\pi\, d = \frac{R_{\max}\, L_\pi}{(1-\gamma)^2}\, d.
\]

\emph{Step 4: Contradiction.} Substituting both bounds into Equation~\ref{eq:mono_holder}:
\[
    \eta\, d \leq L_V\, d + \frac{R_{\max}\, L_\pi}{(1-\gamma)^2}\, d.
\]
Dividing by $d > 0$:
\[
    \eta \leq L_V + \frac{R_{\max}\, L_\pi}{(1-\gamma)^2}.
\]
This contradicts the condition in Equation~\ref{eq:contraction_condition_main}. Therefore $d = 0$, and the TMF equilibrium is unique. $\hfill\square$
\section{Proof of Theorem~\ref{thm:n_approx}}
\label{app:proof_napprox}

\noindent\textbf{Theorem~\ref{thm:n_approx}} ($N$-agent approximation).\quad
\emph{Under Assumptions~\ref{asmp:lipschitz}--\ref{asmp:monotonicity}, let $(\pi^*, \{\mu_t^*\})$ be the TMF equilibrium. For any batch size $B$ such that $\alpha = \tfrac{B}{N}(\rho_P + L_P) + (1 - \tfrac{B}{N})(\rho_{P_0} + L_{P_0}) < 1$, the exploitability}
\[
    \epsilon_N = \max_i \sup_{\pi^i} \bigl[V_{\pi^i}(o_0^i, \hat{\mu}^N_{\pi^i, \pi^{*,-i}}) - V_{\pi^*}(o_0^i, \hat{\mu}^N_{\pi^*})\bigr]
    \leq \frac{C}{\sqrt{N}},
\]
\emph{where $C$ depends on $L_r$, $L_P$, $L_{P_0}$, $\rho_P$, $\rho_{P_0}$, $|\cO|$, and $T$, but not on $B$.}

\noindent\textit{Proof.}\quad The proof follows the propagation-of-chaos framework (\citealp{carmona2018probabilistic}). Unlike the standard synchronous setting, where all $N$ agents update simultaneously in $T$ steps, asynchronous batching produces $\tau = NT/B$ stages, each updating $B$ agents while the remaining $N-B$ transition passively via $P_0$. We compare two population trajectories:
\begin{itemize}
    \item $\mu_t^*$: the \emph{planned} trajectory, computed deterministically by the forward model (Equation~\ref{eq:mu_expected}) under the equilibrium policy $\pi^*$;
    \item $\hat{\mu}_t^N$: the \emph{realized} trajectory, the stochastic empirical distribution when $N$ agents actually execute $\pi^*$, i.e.\ $\hat{\mu}_t^N(o) = \frac{1}{N}\sum_{i=1}^N \mathbbm{1}[o_t^i = o]$.
\end{itemize}
The proof proceeds in five steps: reduce exploitability to the gap $\Wass(\hat{\mu}_t^N, \mu_t^*)$ (Step~0), bound the initial gap (Step~1), analyze per-step error propagation (Step~2), accumulate errors over all time steps (Step~3), and convert to an exploitability bound (Step~4).

\paragraph{Step 0: From exploitability to mean-field gap.} We show that bounding $\epsilon_N$ reduces to bounding $\Wass(\hat{\mu}_t^N, \mu_t^*)$. Consider any unilateral deviation where agent $i$ switches from $\pi^*$ to some alternative $\pi^i$, and write $\hat{\mu}^{N,\mathrm{dev}}$ for the empirical trajectory that results (it differs from $\hat{\mu}^N$ only in agent $i$'s states). We decompose the deviation gain by inserting the planned trajectory $\mu^*$ as a reference:
\begin{align}
    V_{\pi^i}(o, \hat{\mu}^{N,\mathrm{dev}}) - V_{\pi^*}(o, \hat{\mu}^N)
    &= \underbrace{\bigl[V_{\pi^i}(o, \hat{\mu}^{N,\mathrm{dev}}) - V_{\pi^i}(o, \mu^*)\bigr]}_{\text{(a)}}
     + \underbrace{\bigl[V_{\pi^i}(o, \mu^*) - V_{\pi^*}(o, \mu^*)\bigr]}_{\text{(b)}}
     + \underbrace{\bigl[V_{\pi^*}(o, \mu^*) - V_{\pi^*}(o, \hat{\mu}^N)\bigr]}_{\text{(c)}}.
\end{align}
We bound each term separately.

\emph{Term~(a).} Agent $i$'s deviation changes exactly one agent's state out of $N$, so the empirical distribution shifts by at most $1/N$ in Wasserstein distance at each step:
\[
    \max_t \Wass(\hat{\mu}_t^{N,\mathrm{dev}}, \mu_t^*) \leq \max_t \Wass(\hat{\mu}_t^N, \mu_t^*) + O(1/N).
\]
Applying Lemma~\ref{lem:value_sensitivity} (Equation~\ref{eq:value_sensitivity}):
\[
    |V_{\pi^i}(o, \hat{\mu}^{N,\mathrm{dev}}) - V_{\pi^i}(o, \mu^*)| \leq L_V \cdot \bigl(\max_t \Wass(\hat{\mu}_t^N, \mu_t^*) + O(1/N)\bigr).
\]

\emph{Term~(b).} Since $\pi^*$ is the equilibrium policy for the mean-field game under trajectory $\mu^*$, no deviation can improve the value:
\[
    V_{\pi^i}(o, \mu^*) - V_{\pi^*}(o, \mu^*) \leq 0.
\]

\emph{Term~(c).} Applying Lemma~\ref{lem:value_sensitivity} again:
\[
    |V_{\pi^*}(o, \mu^*) - V_{\pi^*}(o, \hat{\mu}^N)| \leq L_V \cdot \max_t \Wass(\hat{\mu}_t^N, \mu_t^*).
\]

\emph{Combining (a)--(c):}
\begin{equation}
    \epsilon_N \leq 2L_V \cdot \EE\bigl[\max_t \Wass(\hat{\mu}_t^N, \mu_t^*)\bigr] + O(L_V/N),
    \label{eq:exploit_reduction}
\end{equation}
so bounding $\EE[\max_t \Wass(\hat{\mu}_t^N, \mu_t^*)]$ directly yields a bound on $\epsilon_N$.

\paragraph{Step 1: Initial error.} At $t=0$, each agent's state is drawn i.i.d.\ from $\mu_0$, so $\hat{\mu}_0^N$ is the empirical distribution of $N$ independent samples. For each state $o$, $\hat{\mu}_0^N(o)$ is a sample mean with variance $\mu_0(o)(1-\mu_0(o))/N \leq 1/(4N)$ (since $p(1-p) \leq 1/4$ for all $p \in [0,1]$). By Jensen's inequality ($\EE[|X|] \leq \sqrt{\EE[X^2]}$):
\[
    \EE\bigl[|\hat{\mu}_0^N(o) - \mu_0(o)|\bigr] \leq \sqrt{\mathrm{Var}[\hat{\mu}_0^N(o)]} \leq \frac{1}{2\sqrt{N}}.
\]
Summing over all $o \in \cO$:
\begin{equation}
    \EE[\Wass(\hat{\mu}_0^N, \mu_0)] = \EE\Bigl[\sum_o |\hat{\mu}_0^N(o) - \mu_0(o)|\Bigr] \leq \frac{|\cO|}{2\sqrt{N}}.
    \label{eq:initial_error}
\end{equation}

\paragraph{Step 2: Per-step error propagation.} We now bound how the gap $\Wass(\hat{\mu}_t^N, \mu_t^*)$ evolves from step $t$ to $t+1$. The \emph{planned} trajectory updates deterministically via the forward model:
\[
    \mu_{t+1}^*(o') = \sum_o \mu_t^*(o) \Bigl[\tfrac{B}{N}\sum_a \pi^*(a|o)\, P(o'|o,a,\mu_t^*) + \bigl(1-\tfrac{B}{N}\bigr) P_0(o'|o,\mu_t^*)\Bigr].
\]
In the $N$-agent system, the \emph{realized} trajectory instead evolves stochastically:
\[
    \hat{\mu}_{t+1}^N(o') = \frac{1}{N}\sum_{i=1}^N \mathbbm{1}[o_{t+1}^i = o'],
\]
where each agent $i$ transitions using the \emph{realized} distribution $\hat{\mu}_t^N$ in place of the planned $\mu_t^*$: active agents transition via $P(\cdot|o_t^i, a_t^i, \hat{\mu}_t^N)$ and passive agents via $P_0(\cdot|o_t^i, \hat{\mu}_t^N)$.

Comparing these two updates, we decompose the gap at step $t+1$ via the triangle inequality, inserting the conditional mean $\EE[\hat{\mu}_{t+1}^N \mid \hat{\mu}_t^N]$ as a reference:
\begin{align}
    \EE[\Wass(\hat{\mu}_{t+1}^N, \mu_{t+1}^*) \mid \hat{\mu}_t^N]
    &\leq \underbrace{\EE[\Wass(\hat{\mu}_{t+1}^N, \EE[\hat{\mu}_{t+1}^N \mid \hat{\mu}_t^N]) \mid \hat{\mu}_t^N]}_{\text{Noise}} + \underbrace{\Wass(\EE[\hat{\mu}_{t+1}^N \mid \hat{\mu}_t^N], \mu_{t+1}^*)}_{\text{Bias}}.
    \label{eq:bias_noise_decomp}
\end{align}
We now bound each term separately.

\emph{Bias (second term).} The conditional mean $\EE[\hat{\mu}_{t+1}^N \mid \hat{\mu}_t^N]$ and the planned trajectory $\mu_{t+1}^*$ differ for two reasons: the transition parameters are evaluated at $\hat{\mu}_t^N$ instead of $\mu_t^*$, and the starting distributions $\hat{\mu}_t^N \neq \mu_t^*$ are different. For the parameter perturbation, the Lipschitz condition on $P$ (Assumption~\ref{asmp:lipschitz}) gives, for each active agent:
\[
    \Wass\bigl(P(\cdot|o, a, \hat{\mu}_t^N),\; P(\cdot|o, a, \mu_t^*)\bigr) \leq L_P\, \Wass(\hat{\mu}_t^N, \mu_t^*).
\]
Similarly, for each passive agent $i$:
\[
    \Wass\bigl(P_0(\cdot|o, \hat{\mu}_t^N),\; P_0(\cdot|o, \mu_t^*)\bigr) \leq L_{P_0}\, \Wass(\hat{\mu}_t^N, \mu_t^*).
\]
For the starting distribution difference, the Dobrushin contraction coefficients $\rho_P$, $\rho_{P_0}$ (Assumption~\ref{asmp:ergodicity}) bound the propagation of the gap $\hat{\mu}_t^N \neq \mu_t^*$. To separate the two contributions via the triangle inequality, define
\[
    \bar{\mu}_{t+1}(o') = \sum_o \hat{\mu}_t^N(o)\,\Bigl[\tfrac{B}{N}\textstyle\sum_a \pi^*(a|o)\,P(o'|o,a,\mu_t^*) + \tfrac{N-B}{N}\,P_0(o'|o,\mu_t^*)\Bigr],
\]
which applies the planned transition parameters ($\mu_t^*$) to the realized starting distribution ($\hat{\mu}_t^N$). By the triangle inequality:
\begin{align}
    \Wass\bigl(\EE[\hat{\mu}_{t+1}^N \mid \hat{\mu}_t^N],\; \mu_{t+1}^*\bigr)
    &\leq \Wass(\bar{\mu}_{t+1},\; \mu_{t+1}^*) + \Wass\bigl(\EE[\hat{\mu}_{t+1}^N \mid \hat{\mu}_t^N],\; \bar{\mu}_{t+1}\bigr).
    \label{eq:bias_triangle}
\end{align}
The first term compares two distributions obtained by applying the \emph{same} transition (with parameters $\mu_t^*$) to different starting distributions $\hat{\mu}_t^N$ and $\mu_t^*$. By the Dobrushin contraction property (Assumption~\ref{asmp:ergodicity}):
\[
    \Wass(\bar{\mu}_{t+1},\; \mu_{t+1}^*) \leq \bigl(\tfrac{B}{N}\,\rho_P + \tfrac{N-B}{N}\,\rho_{P_0}\bigr)\, \Wass(\hat{\mu}_t^N, \mu_t^*).
\]
The second term compares two distributions with the \emph{same} starting distribution $\hat{\mu}_t^N$ but different transition parameters ($\hat{\mu}_t^N$ vs.\ $\mu_t^*$). By the per-agent Lipschitz bounds above:
\[
    \Wass\bigl(\EE[\hat{\mu}_{t+1}^N \mid \hat{\mu}_t^N],\; \bar{\mu}_{t+1}\bigr) \leq \bigl(\tfrac{B}{N}\,L_P + \tfrac{N-B}{N}\,L_{P_0}\bigr)\, \Wass(\hat{\mu}_t^N, \mu_t^*).
\]
Combining both terms:
\begin{equation}
    \Wass\bigl(\EE[\hat{\mu}_{t+1}^N \mid \hat{\mu}_t^N],\; \mu_{t+1}^*\bigr)
    \leq \bigl[\tfrac{B}{N}(\rho_P + L_P) + \tfrac{N-B}{N}(\rho_{P_0} + L_{P_0})\bigr]\, \Wass(\hat{\mu}_t^N, \mu_t^*).
    \label{eq:bias_bound}
\end{equation}

\emph{Noise (first term).} Conditionally on $\hat{\mu}_t^N$, the $N$ agents transition independently. The conditional variance of $\hat{\mu}_{t+1}^N(o')$ for any $o'$ is:
\begin{equation}
    \mathrm{Var}\bigl[\hat{\mu}_{t+1}^N(o') \mid \hat{\mu}_t^N\bigr] = \frac{1}{N^2}\sum_{i=1}^N \mathrm{Var}\bigl[\mathbbm{1}[o_{t+1}^i = o'] \mid o_t^i, \hat{\mu}_t^N\bigr] \leq \frac{1}{4N},
    \label{eq:variance_bound}
\end{equation}
since each indicator $\mathbbm{1}[o_{t+1}^i = o']$ is Bernoulli with parameter $p_i \in [0,1]$ and variance $p_i(1-p_i) \leq 1/4$. For each $o'$, applying Jensen's inequality:
\[
    \EE\bigl[|\hat{\mu}_{t+1}^N(o') - \EE[\hat{\mu}_{t+1}^N(o') \mid \hat{\mu}_t^N]| \;\big|\; \hat{\mu}_t^N\bigr] \leq \sqrt{\mathrm{Var}[\hat{\mu}_{t+1}^N(o') \mid \hat{\mu}_t^N]} \leq \frac{1}{2\sqrt{N}}.
\]
Summing over all $o' \in \cO$, we obtain the per-step noise bound on $\Wass$:
\begin{align}
    \EE\bigl[\Wass(\hat{\mu}_{t+1}^N, \EE[\hat{\mu}_{t+1}^N \mid \hat{\mu}_t^N]) \mid \hat{\mu}_t^N\bigr]
    &= \sum_{o' \in \cO} \EE\bigl[|\hat{\mu}_{t+1}^N(o') - \EE[\hat{\mu}_{t+1}^N(o') \mid \hat{\mu}_t^N]| \;\big|\; \hat{\mu}_t^N\bigr] \notag \\
    &\leq \sum_{o' \in \cO} \frac{1}{2\sqrt{N}} = \frac{|\cO|}{2\sqrt{N}}.
    \label{eq:noise_bound}
\end{align}

\emph{Combined recursion.} Substituting Equation~\ref{eq:bias_bound} and Equation~\ref{eq:noise_bound} into Equation~\ref{eq:bias_noise_decomp}:
\begin{equation}
    \EE[\Wass(\hat{\mu}_{t+1}^N, \mu_{t+1}^*) \mid \hat{\mu}_t^N] \leq \bigl[\tfrac{B}{N}(\rho_P + L_P) + \tfrac{N-B}{N}(\rho_{P_0} + L_{P_0})\bigr]\, \Wass(\hat{\mu}_t^N, \mu_t^*) + \frac{|\cO|}{2\sqrt{N}}.
    \label{eq:combined_recursion}
\end{equation}

\paragraph{Step 3: Error accumulation.} Denote $\alpha = \tfrac{B}{N}(\rho_P + L_P) + \tfrac{N-B}{N}(\rho_{P_0} + L_{P_0})$ (the coefficient in Equation~\ref{eq:combined_recursion}) and let $e_t = \EE[\Wass(\hat{\mu}_t^N, \mu_t^*)]$. Taking unconditional expectations in Equation~\ref{eq:combined_recursion}:
\begin{equation}
    e_{t+1} \leq \alpha\, e_t + \frac{|\cO|}{2\sqrt{N}}.
    \label{eq:et_recursion}
\end{equation}
Iterating $t$ times:
\begin{equation}
    e_t \leq \alpha^t\, e_0 + \frac{|\cO|}{2\sqrt{N}} \sum_{j=0}^{t-1} \alpha^j
    = \alpha^t\, e_0 + \frac{|\cO|}{2\sqrt{N}} \cdot \frac{1 - \alpha^t}{1 - \alpha}.
    \label{eq:iteration}
\end{equation}
The condition $\alpha < 1$ in Theorem~\ref{thm:n_approx} ensures that the geometric sum converges. When $\rho_P + L_P < \rho_{P_0} + L_{P_0}$ (active transitions contract more), solving $\alpha < 1$ for $B/N$ yields the equivalent lower bound
\[
    \frac{B}{N} > \frac{\rho_{P_0} + L_{P_0} - 1}{(\rho_{P_0} + L_{P_0}) - (\rho_P + L_P)},
\]
as stated in Remark~\ref{rem:tradeoff}. Under either form, $\frac{1}{1-\alpha}$ is a finite constant independent of $T$.

Since $\alpha < 1$, Equation~\ref{eq:iteration} gives a uniform bound over all $t$: $\alpha^t \leq 1$ and $\frac{1-\alpha^t}{1-\alpha} \leq \frac{1}{1-\alpha}$, so
\[
    e_t \leq e_0 + \frac{|\cO|}{2\sqrt{N}(1-\alpha)}, \quad \text{for all } t \geq 0.
\]
Substituting $e_0 \leq \frac{|\cO|}{2\sqrt{N}}$ from Step~1 (Equation~\ref{eq:initial_error}):
\begin{equation}
    \sup_{t \geq 0}\, e_t \leq \frac{|\cO|}{2\sqrt{N}} + \frac{|\cO|}{2\sqrt{N}(1-\alpha)} = \frac{|\cO|}{2\sqrt{N}} \cdot \frac{2-\alpha}{1-\alpha} \leq \frac{C}{\sqrt{N}},
    \label{eq:sup_bound}
\end{equation}
where $C = \frac{|\cO|(2-\alpha)}{2(1-\alpha)}$ depends on $\rho_P$, $L_P$, $\rho_{P_0}$, $L_{P_0}$, and $|\cO|$, but not on $T$.

\paragraph{Step 4: Exploitability bound.}
To complete the proof, we need to bound $\EE[\max_t \Wass(\hat{\mu}_t^N, \mu_t^*)]$ in Equation~\ref{eq:exploit_reduction}. Since the exploitability involves at most $T+1$ decision periods, we can pass from $\sup_t \EE[\cdot]$ in Equation~\ref{eq:sup_bound} to $\EE[\max_t (\cdot)]$:
\[
    \EE\bigl[\max_{0 \leq t \leq T} \Wass(\hat{\mu}_t^N, \mu_t^*)\bigr] \leq \sum_{t=0}^{T} \EE[\Wass(\hat{\mu}_t^N, \mu_t^*)] \leq (T+1)\, \sup_t e_t \leq \frac{(T+1)\, C}{\sqrt{N}}.
\]
This bound is loose: it replaces $\EE[\max_t]$ with $(T+1)\,\sup_t \EE[\cdot]$, but since both quantities are $O(1/\sqrt{N})$, the looseness only affects the constant $C'$, not the scaling in $N$. Substituting this bound into Equation~\ref{eq:exploit_reduction}:
\begin{equation}
    \epsilon_N \leq 2L_V \cdot \frac{(T+1)\, C}{\sqrt{N}} + O(L_V/N) \leq \frac{C'}{\sqrt{N}},
\end{equation}
where the constant $C'$ depends on $L_r$, $L_P$, $L_{P_0}$, $|\cO|$, and $T$, but not on $B$. Since the bound holds for each fixed $B$ with the same constant $C'$, it holds uniformly over all $B \in \{1, \ldots, N\}$. In particular, the result extends to time-varying batch sizes $B_t$, as each stage's contribution to the recursion depends only on the Lipschitz constants and not on $B_t$. $\hfill\square$

\section{Proof of Theorem~\ref{thm:convergence}}
\label{app:proof_convergence}

\noindent\textbf{Theorem~\ref{thm:convergence}} (Convergence of TMF-PG).\quad
\emph{Under the conditions of Theorem~\ref{thm:existence_uniqueness}, suppose the backward pass at each iteration $k$ produces a policy $\pi_k$ with $\|\pi_k - \pi_{\boldsymbol{\mu}^{(k)}}\|_\infty \leq \epsilon_k$ and $\epsilon_k \to 0$. Then $\max_t \Wass(\mu_t^{(k)}, \mu_t^*) \to 0$ and $\|\pi_k - \pi^*\|_\infty \to 0$. The convergence holds for any $B$ satisfying the condition of Theorem~\ref{thm:n_approx}.}

\noindent\textit{Proof.}\quad We work in the large-population limit, where the rollout concentrates around the deterministic forward map $\mathrm{Fwd}$, yielding a deterministic trajectory $\boldsymbol{\mu}^{(k)} = (\mu_0^{(k)}, \mu_1^{(k)}, \ldots, \mu_T^{(k)})$ rather than a stochastic empirical distribution. The gap between the deterministic trajectory and the finite-$N$ empirical distribution is bounded by Theorem~\ref{thm:n_approx}. In this deterministic setting, each TMF-PG iteration approximately applies the operator $\Gamma = \mathrm{Fwd} \circ \mathrm{BR}$: the policy update produces an approximate best response $\pi_k \approx \mathrm{BR}(\boldsymbol{\mu}^{(k)})$ with error $\epsilon_k \to 0$, and the rollout produces the next trajectory $\boldsymbol{\mu}^{(k+1)} = \mathrm{Fwd}(\pi_k)$. We show that these approximate iterations converge to the unique TMF equilibrium $\boldsymbol{\mu}^*$. The proof proceeds in three steps.

\textbf{Step 1: $\Gamma$ is continuous.} Recall that $\Gamma = \mathrm{Fwd} \circ \mathrm{BR}$, where $\mathrm{BR}$ is $L_\pi$-Lipschitz by Assumption~\ref{asmp:lipschitz}. Since $\mathrm{BR}$ is Lipschitz, it suffices to show that $\mathrm{Fwd}$ is Lipschitz. Consider two policies $\pi$, $\pi'$ with $\|\pi - \pi'\|_\infty \leq \epsilon$, and let $\boldsymbol{\mu} = \mathrm{Fwd}(\pi)$, $\boldsymbol{\mu}' = \mathrm{Fwd}(\pi')$. Both trajectories share the same initial condition $\mu_0 = \mu_0'$, so $\Wass(\mu_0, \mu_0') = 0$. We derive a per-step recursion bounding $\Wass(\mu_{t+1}, \mu_{t+1}')$ in terms of $\Wass(\mu_t, \mu_t')$ and $\epsilon$, then unroll it to obtain $\max_t \Wass(\mu_t, \mu_t') \leq L_F\, \epsilon$.

\emph{Per-step decomposition.} To bound the gap $\Wass(\mu_{t+1}, \mu_{t+1}')$ at step $t$, we introduce an intermediate distribution $\tilde{\mu}_{t+1}$, defined as the result of evolving $\mu_t'$ forward one step under $\pi$ (rather than $\pi'$). This separates the gap into two sources: (a) the \emph{inherited gap} from different starting distributions $\mu_t$ vs.\ $\mu_t'$, and (b) the \emph{policy discrepancy} from using $\pi$ vs.\ $\pi'$ on the same $\mu_t'$. By the triangle inequality:
\[
    \Wass(\mu_{t+1}, \mu_{t+1}')
    \leq \underbrace{\Wass(\mu_{t+1}, \tilde{\mu}_{t+1})}_{\text{inherited gap (contracts)}}
    + \underbrace{\Wass(\tilde{\mu}_{t+1}, \mu_{t+1}')}_{\text{policy discrepancy}}.
\]

\emph{Inherited gap (first term).} Both $\mu_{t+1}$ and $\tilde{\mu}_{t+1}$ are evolved under the same policy $\pi$, but from different distributions $\mu_t$ and $\mu_t'$. By the same per-agent coupling as in the bias bound (Equation~\ref{eq:bias_bound}):
\begin{equation}
    \Wass(\mu_{t+1}, \tilde{\mu}_{t+1}) \leq \alpha\, \Wass(\mu_t, \mu_t'),
    \label{eq:inherited_gap}
\end{equation}
where $\alpha = \tfrac{B}{N}(\rho_P + L_P) + \tfrac{N-B}{N}(\rho_{P_0} + L_{P_0}) < 1$ is the contraction coefficient from Equation~\ref{eq:combined_recursion}.

\emph{Policy discrepancy (second term).} Both $\tilde{\mu}_{t+1}$ and $\mu_{t+1}'$ are evolved from the same $\mu_t'$, but under $\pi$ vs.\ $\pi'$. Expanding both using the forward model (Equation~\ref{eq:mu_expected}):
\begin{align*}
    \tilde{\mu}_{t+1}(o') &= \sum_o \mu_t'(o) \Bigl[\tfrac{B}{N} \textstyle\sum_a \pi(a|o,\mu_t')\, P(o'|o,a,\mu_t') + \tfrac{N-B}{N}\, P_0(o'|o,\mu_t')\Bigr], \\
    \mu_{t+1}'(o') &= \sum_o \mu_t'(o) \Bigl[\tfrac{B}{N} \textstyle\sum_a \pi'(a|o,\mu_t')\, P(o'|o,a,\mu_t') + \tfrac{N-B}{N}\, P_0(o'|o,\mu_t')\Bigr].
\end{align*}
Since both use the same base distribution $\mu_t'$ and the same population state $\mu_t'$ in both $P$ and $P_0$, the passive terms $\tfrac{N-B}{N} P_0(o'|o,\mu_t')$ are identical ($P_0$ does not depend on the policy). Subtracting:
\[
    \tilde{\mu}_{t+1}(o') - \mu_{t+1}'(o')
    = \frac{B}{N} \sum_o \mu_t'(o) \sum_a \bigl[\pi(a|o,\mu_t') - \pi'(a|o,\mu_t')\bigr]\, P(o'|o,a,\mu_t').
\]
Taking the $\ell_1$ norm over $o'$, since $\mu_t'(o) \geq 0$ and $P(o'|o,a,\mu_t') \geq 0$, the triangle inequality moves the absolute value inside:
\begin{align}
    \Wass(\tilde{\mu}_{t+1}, \mu_{t+1}')
    &= \sum_{o'} \bigl|\tilde{\mu}_{t+1}(o') - \mu_{t+1}'(o')\bigr| \nonumber\\
    &\leq \frac{B}{N} \sum_{o'} \sum_o \mu_t'(o) \sum_a |\pi(a|o,\mu_t') - \pi'(a|o,\mu_t')|\, P(o'|o,a,\mu_t') \nonumber\\
    &= \frac{B}{N} \sum_o \mu_t'(o) \sum_a |\pi(a|o,\mu_t') - \pi'(a|o,\mu_t')|\, \underbrace{\sum_{o'} P(o'|o,a,\mu_t')}_{=\,1} \nonumber\\
    &= \frac{B}{N} \sum_o \mu_t'(o) \sum_a |\pi(a|o,\mu_t') - \pi'(a|o,\mu_t')| \nonumber\\
    &\leq \frac{B}{N} \underbrace{\sum_o \mu_t'(o)}_{=\,1}\; |\cA|\, \|\pi - \pi'\|_\infty
    = \frac{B\,|\cA|}{N}\, \epsilon.
    \label{eq:policy_discrepancy}
\end{align}

\emph{Combining and unrolling.} Adding the inherited gap (Equation~\ref{eq:inherited_gap}) and the policy discrepancy (Equation~\ref{eq:policy_discrepancy}) via the triangle inequality gives the per-step recursion:
\begin{equation}
    \Wass(\mu_{t+1}, \mu_{t+1}') \leq \alpha\, \Wass(\mu_t, \mu_t') + \tfrac{B\,|\cA|}{N}\, \epsilon.
    \label{eq:fwd_lipschitz_step}
\end{equation}
Since $\Wass(\mu_0, \mu_0') = 0$, unrolling this recursion $t$ times yields:
\[
    \Wass(\mu_t, \mu_t') \leq \tfrac{B\,|\cA|}{N}\, \epsilon \sum_{s=0}^{t-1} \alpha^s \quad \text{for each } t.
\]
Taking the maximum over $t = 0, \ldots, T$ and bounding the geometric sum:
\begin{equation}
    \max_{0 \leq t \leq T}\, \Wass(\mu_t, \mu_t')
    \leq \frac{B\,|\cA|}{N}\, \epsilon \cdot \sum_{s=0}^{T-1} \alpha^s
    = \frac{B\,|\cA|}{N} \cdot \frac{1 - \alpha^T}{1 - \alpha}\, \epsilon
    \leq \frac{B\,|\cA|}{N(1-\alpha)}\, \epsilon
    =: L_F\, \epsilon.
    \label{eq:fwd_lipschitz}
\end{equation}
This establishes that $\mathrm{Fwd}$ is $L_F$-Lipschitz. Since $\Gamma = \mathrm{Fwd} \circ \mathrm{BR}$ is a composition of two Lipschitz maps ($\mathrm{Fwd}$ with constant $L_F$ and $\mathrm{BR}$ with constant $L_\pi$), $\Gamma$ is Lipschitz with constant $L_\Gamma = L_F\, L_\pi$, and in particular continuous on the compact set $\Delta(\cO)^T$.


\textbf{Step 2: Trajectory convergence.} Let $d_k := \max_t \Wass(\mu_t^{(k)}, \mu_t^*)$ denote the gap between the $k$-th iterate and the equilibrium trajectory. We show $d_k \to 0$ by establishing a contracting recursion across iterations.

Since $\boldsymbol{\mu}^{(k+1)} = \mathrm{Fwd}(\pi_k)$ and $\boldsymbol{\mu}^* = \mathrm{Fwd}(\mathrm{BR}(\boldsymbol{\mu}^*)) = \Gamma(\boldsymbol{\mu}^*)$ (by definition of the TMF equilibrium), Step~1 gives
\begin{equation}
    d_{k+1} = \max_t \Wass\bigl(\mathrm{Fwd}(\pi_k)_t,\, \mathrm{Fwd}(\mathrm{BR}(\boldsymbol{\mu}^*))_t\bigr) \leq L_F\, \|\pi_k - \mathrm{BR}(\boldsymbol{\mu}^*)\|_\infty.
    \label{eq:dk_fwd}
\end{equation}
By the triangle inequality and the $L_\pi$-Lipschitz property of $\mathrm{BR}$ (Assumption~\ref{asmp:lipschitz}):
\begin{equation}
    \|\pi_k - \mathrm{BR}(\boldsymbol{\mu}^*)\|_\infty
    \leq \underbrace{\|\pi_k - \mathrm{BR}(\boldsymbol{\mu}^{(k)})\|_\infty}_{\leq\, \epsilon_k}
    + \underbrace{\|\mathrm{BR}(\boldsymbol{\mu}^{(k)}) - \mathrm{BR}(\boldsymbol{\mu}^*)\|_\infty}_{\leq\, L_\pi\, d_k}.
    \label{eq:pi_triangle}
\end{equation}
Substituting Equation~\ref{eq:pi_triangle} into Equation~\ref{eq:dk_fwd}:
\begin{equation}
    d_{k+1} \leq L_F\, (\epsilon_k + L_\pi\, d_k).
    \label{eq:dk_expanded}
\end{equation}
Recalling $L_F = \frac{B\,|\cA|}{N(1-\alpha)}$ from Equation~\ref{eq:fwd_lipschitz} and defining $L_\Gamma := L_F\, L_\pi$:
\[
    L_\Gamma = \frac{B\,|\cA|\, L_\pi}{N(1-\alpha)},
\]
so $L_\Gamma < 1$ whenever $N > \frac{B\,|\cA|\, L_\pi}{1-\alpha}$, which holds in the large-population regime. Equation~\ref{eq:dk_expanded} then becomes the contracting recursion:
\begin{equation}
    d_{k+1} \leq L_\Gamma\, d_k + L_F\, \epsilon_k.
    \label{eq:dk_recursion}
\end{equation}
Unrolling by induction on $k$ gives
\begin{equation}
    d_k \leq L_\Gamma^k\, d_0 + L_F \sum_{j=0}^{k-1} L_\Gamma^{k-1-j}\, \epsilon_j.
    \label{eq:dk_unrolled}
\end{equation}
The first term satisfies $L_\Gamma^k\, d_0 \to 0$ since $L_\Gamma < 1$. It remains to show that the second term in Equation~\ref{eq:dk_unrolled} also vanishes. Fix any $\delta > 0$; we show that the second term is eventually less than $\delta$.

Since $\epsilon_j \to 0$ by assumption, we can choose $K$ large enough that $\epsilon_j$ is small for all $j \geq K$. Specifically, choose $K$ such that
\begin{equation}
    \epsilon_j \leq \frac{\delta(1 - L_\Gamma)}{2\, L_F} \quad \text{for all } j \geq K.
    \label{eq:K_choice}
\end{equation}
Now split the sum at $j = K$. For the terms $j \geq K$, substitute Equation~\ref{eq:K_choice} and apply the geometric series formula $\sum_{s=0}^{\infty} r^s = \frac{1}{1-r}$ (with $r = L_\Gamma < 1$):
\begin{align}
    L_F \sum_{j=K}^{k-1} L_\Gamma^{k-1-j}\, \epsilon_j
    &\leq \frac{\delta(1-L_\Gamma)}{2} \sum_{j=K}^{k-1} L_\Gamma^{k-1-j} \notag\\
    &\leq \frac{\delta(1-L_\Gamma)}{2} \cdot \frac{1}{1-L_\Gamma}
    = \frac{\delta}{2}.
    \label{eq:late_bound}
\end{align}
For the terms $j < K$, there are only finitely many (at most $K$), and each carries the factor $L_\Gamma^{k-1-j}$. Since $j \leq K-1$ implies $k-1-j \geq k-K$, the entire sum can be bounded by $L_\Gamma^{k-K}$ times a finite constant independent of $k$:
\begin{align}
    L_F \sum_{j=0}^{K-1} L_\Gamma^{k-1-j}\, \epsilon_j
    &= L_\Gamma^{k-K} \cdot L_F \sum_{j=0}^{K-1} L_\Gamma^{K-1-j}\, \epsilon_j \notag\\
    &\leq L_\Gamma^{k-K} \cdot \frac{L_F\, \max_{j < K} \epsilon_j}{1-L_\Gamma}.
    \label{eq:early_bound}
\end{align}
Since $L_\Gamma^{k-K} \to 0$ as $k \to \infty$, there exists $k_0$ such that Equation~\ref{eq:early_bound} is also less than $\delta/2$ for all $k \geq k_0$. Substituting Equation~\ref{eq:late_bound} and Equation~\ref{eq:early_bound} into Equation~\ref{eq:dk_unrolled}:
\[
    d_k \leq L_\Gamma^k\, d_0 + \delta \quad \text{for all } k \geq k_0.
\]
Taking $k \to \infty$ and using $L_\Gamma^k\, d_0 \to 0$:
\[
    \limsup_{k \to \infty}\, d_k \leq \delta.
\]
Since $\delta > 0$ was arbitrary, $d_k \to 0$.

\textbf{Step 3: Policy convergence.} It remains to show $\|\pi_k - \pi^*\|_\infty \to 0$. Since $\pi^* = \mathrm{BR}(\boldsymbol{\mu}^*)$ by the TMF equilibrium definition (Definition~\ref{def:mfe}), the triangle inequality gives:
\[
    \|\pi_k - \pi^*\|_\infty
    \leq \|\pi_k - \mathrm{BR}(\boldsymbol{\mu}^{(k)})\|_\infty
    + \|\mathrm{BR}(\boldsymbol{\mu}^{(k)}) - \mathrm{BR}(\boldsymbol{\mu}^*)\|_\infty.
\]
The first term is at most $\epsilon_k \to 0$ by the theorem assumption on the backward pass. The second term is at most $L_\pi\, d_k$ by the $L_\pi$-Lipschitz property of $\mathrm{BR}$ (Assumption~\ref{asmp:lipschitz}), and $d_k \to 0$ by Step~2. Therefore $\|\pi_k - \pi^*\|_\infty \to 0$.

The TMF equilibrium to which the iterates converge is independent of the batch size $B$ (Theorem~\ref{thm:existence_uniqueness}), since the Lipschitz constants $L_P$, $L_{P_0}$, $L_\pi$, and the monotonicity constant $\eta$ are all properties of the game and do not depend on $B$. $\hfill\square$

\begin{remark}
Theorem~\ref{thm:convergence} establishes convergence in the mean-field limit: as the number of iterations $k$ grows, the trajectory $\boldsymbol{\mu}^{(k)}$ converges to the TMF equilibrium trajectory $\boldsymbol{\mu}^*$ and the policy $\pi_k$ converges to $\pi^*$. For a finite system with $N$ agents, Theorem~\ref{thm:n_approx} bounds the gap between the $N$-agent empirical distribution and the mean-field trajectory under any fixed shared policy by $O(1/\sqrt{N})$. Combining these two results, after sufficiently many iterations, the $N$-agent system under the learned policy operates within $O(1/\sqrt{N})$ of the TMF equilibrium in the Wasserstein metric.
\end{remark}

\section{Additional Experimental Results}
\label{app:experiments_detail}

This appendix provides supplementary experimental results referenced in Section~\ref{sec:experiments}.

\subsection{SRSG: Prediction Error across $N \times B$}

Figure~\ref{fig:prediction_heatmap} shows the $L_1$ prediction error for the forward model across the full $N \times B$ grid.

\begin{figure}[H]
    \centering
    \includegraphics[width=0.55\textwidth]{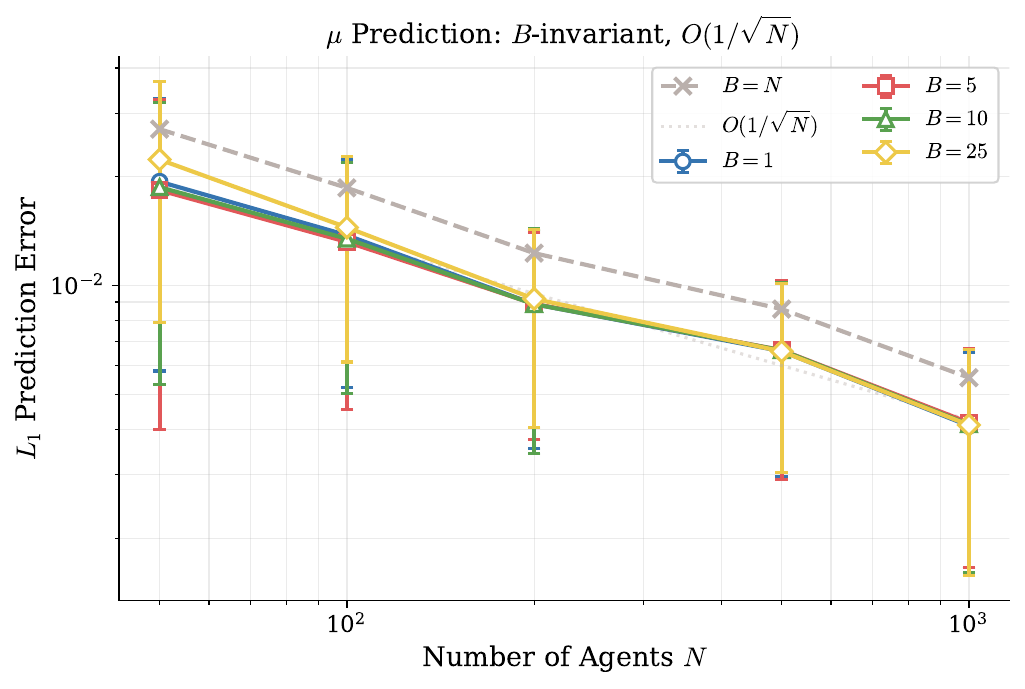}
    \caption{SRSG: $L_1$ forward prediction error across the $N \times B$ grid (40 seeds). For each $N$, error is nearly constant across $B$; for each $B$, error decays as $O(1/\sqrt{N})$ (dashed reference line).}
    \label{fig:prediction_heatmap}
\end{figure}

Two key observations confirm Theorem~\ref{thm:n_approx}:

\begin{enumerate}
    \item \textbf{Stability across $B$:} Fixing $N$ and varying $B$, the prediction error is nearly constant. For example, at $N = 500$, the error is $0.007$ for $B = 1, 5, 10,$ and $25$; only the fully synchronous case ($B = N$) is slightly higher ($0.009$), since a single batch offers no averaging across steps.
    \item \textbf{$O(1/\sqrt{N})$ decay:} Fixing $B$ and varying $N$, the error decreases from $\approx 0.019$ ($N = 50$) to $\approx 0.004$ ($N = 1000$), matching the theoretical rate.
\end{enumerate}

\subsection{DQG: Sensitivity Analysis}

Figure~\ref{fig:dqg_sensitivity} examines TMF-PG's advantage over Myopic under two parameter sweeps, both at $N = 50$: the cliff multiplier $\kappa$, which controls how sharply the service rate drops when a server's queue exceeds its capacity (lower $\kappa$ = harsher penalty), and the horizon $H$, the total number of decision steps.

\begin{figure}[H]
    \centering
    \begin{subfigure}[b]{0.48\textwidth}
        \centering
        \includegraphics[width=\textwidth]{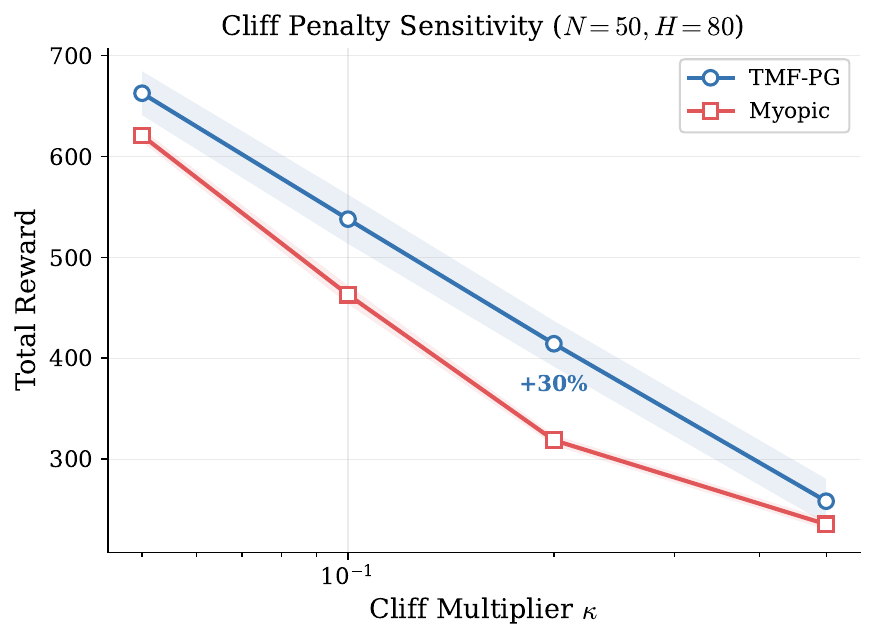}
        \caption{Cliff multiplier sweep ($H = 80$)}
        \label{fig:dqg_sensitivity_a}
    \end{subfigure}
    \hfill
    \begin{subfigure}[b]{0.48\textwidth}
        \centering
        \includegraphics[width=\textwidth]{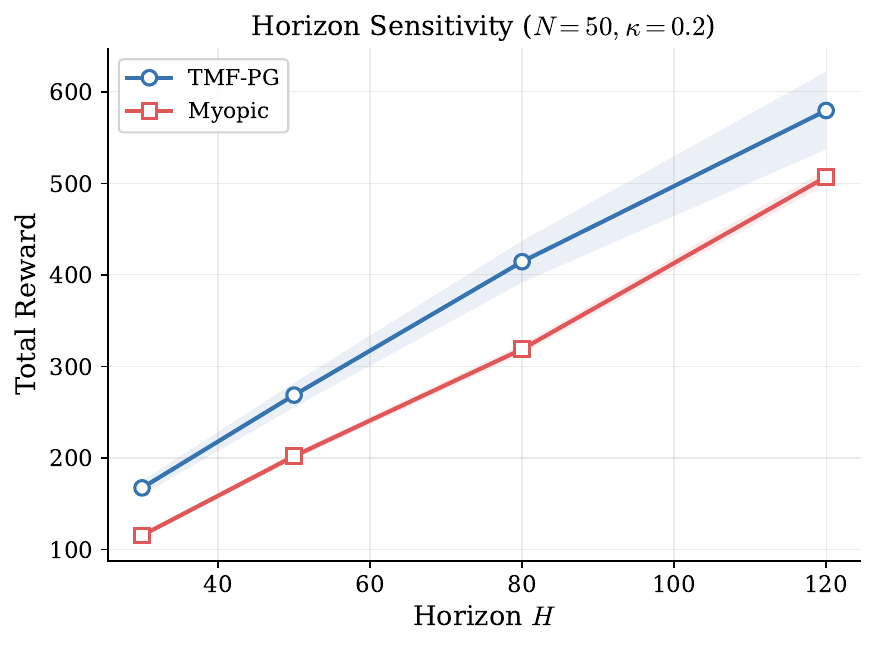}
        \caption{Horizon sweep ($\kappa = 0.2$)}
        \label{fig:dqg_sensitivity_b}
    \end{subfigure}
    \caption{DQG sensitivity analysis ($N = 50$). TMF-PG consistently outperforms Myopic across all tested cliff intensities and horizons.}
    \label{fig:dqg_sensitivity}
\end{figure}

\paragraph{Cliff sensitivity (Figure~\ref{fig:dqg_sensitivity_a}).}
TMF-PG outperforms Myopic across all tested cliff intensities ($\kappa \in [0.05, 0.5]$). The advantage is most pronounced at moderate values ($\kappa = 0.2$: $+30\%$), where the penalty is severe enough to punish naive allocation but the problem retains nontrivial strategic structure. At mild cliff ($0.05$: $+7\%$) the penalty is weak, so even greedy allocation rarely triggers it; at heavy cliff ($0.5$: $+10\%$) the strong gradient aids reactive avoidance. Across the full range, anticipating future arrivals via the forward model consistently yields higher total reward.

\paragraph{Horizon sensitivity (Figure~\ref{fig:dqg_sensitivity_b}).}
TMF-PG's advantage holds across all tested horizons ($H \in [30, 120]$). At $H = 30$ ($+45\%$), the advantage is largest because the short horizon amplifies the cost of early misallocation---a poor initial placement cannot be corrected by subsequent service cycles. As the horizon grows, agent cycling dilutes early mistakes, but TMF-PG retains a meaningful edge: $+30\%$ at $H = 80$ and $+14\%$ at $H = 120$. The diminishing but persistent advantage confirms that forward planning remains valuable even as the mixing time of the queueing dynamics increases.

\end{document}